\documentclass[twocolumn,floatfix,aps,prl,superscriptaddress]{revtex4-2}
\usepackage[T1]{fontenc}
\usepackage{newtxtext}
\usepackage[dvipsnames]{xcolor}
\usepackage{graphicx}%
\usepackage{float}
\usepackage{dcolumn}%
\usepackage{bm}%
\usepackage{amsfonts,amssymb,amscd,amsmath,amsthm}
\usepackage{mathrsfs}
\usepackage[ruled,linesnumbered]{algorithm2e}
\usepackage{algpseudocode}
\usepackage{enumerate}
\usepackage{epsfig}
\usepackage{subfigure}
\usepackage{physics}
\usepackage{makecell}
\usepackage{bbm}
\usepackage{svg}
\usepackage[colorlinks = true,hypertexnames=false]{hyperref}
\hypersetup{colorlinks=true, linkcolor=blue, citecolor=blue, urlcolor=blue}
\usepackage{amsmath, amsthm, amsfonts, amssymb, amsbsy, mathtools, xcolor, bm, bbm, graphicx, changepage, cleveref}

\usepackage[toc,page,header]{appendix}
\usepackage{minitoc}

\usepackage{physics}
\usepackage{epstopdf}
\usepackage{framed}
\usepackage{multirow}
\usepackage{color}
\usepackage{longtable}
\usepackage{comment}
\usepackage{qcircuit}
\usepackage{faktor}
\usepackage[normalem]{ulem}
\usepackage[autostyle]{csquotes}

\usepackage{stmaryrd}
\usepackage{tabularx}
\usepackage{booktabs}
\usepackage{adjustbox}

\usepackage{booktabs}    %
\usepackage{multirow}    %
\usepackage{colortbl}    %
\usepackage{bigdelim}

\makeatletter
\newtheorem*{rep@theorem}{\rep@title}
\newcommand{\newreptheorem}[2]{%
\newenvironment{rep#1}[1]{%
 \def\rep@title{#2 \ref{##1}}%
 \begin{rep@theorem}}%
 {\end{rep@theorem}}}
\makeatother

\newtheorem{theorem}{Theorem}
\newtheorem{proposition}{Proposition}

\newtheorem{lemma}{Lemma}

\newtheorem*{assumption*}{Assumption}
\newtheorem{corollary}{Corollary}
\newtheorem{definition}{Definition}

\usepackage{lineno}

\usepackage{interval}
\intervalconfig{soft open fences}

\newcommand{\diag} {\operatorname{diag}}

\newcommand{\supp} {\operatorname{supp}}

\makeatletter
\pretocmd{\appendix}{\let\addcontentsline\@gobble}{}{}
\makeatother

	\newcommand{\cA}{{\cal A}}  \newcommand{\cC}{{\cal C}}
	\newcommand{\cD}{{\cal D}}  
	\newcommand{\cG}{{\cal G}} \newcommand{\cH}{{\cal H}} 
	 \newcommand{\cK}{{\cal K}} 
	\newcommand{\cM}{{\cal M}}  
	  
	\newcommand{\cS}{{\cal S}} \newcommand{\cT}{{\cal T}} 
	\newcommand{\cV}{{\cal V}} \newcommand{\cW}{{\cal W}}

    \newcommand{\dist}{\mathrm{dist}}
\newcommand{\Herm}{\mathrm{Herm}}

\renewcommand*{\HyperDestNameFilter}[1]{main-#1}

\begin{document}
\makeatletter
\let\mainSectionFormat\@hangfrom@section
\let\mainSectionCounterFormat\@sectioncntformat
\makeatother

\let\oldacl\addcontentsline
\renewcommand{\addcontentsline}[3]{}

\title{Quantum state determinability from local marginals is universally robust}

\author{Wenjun Yu}
\email[]{wenjunyus@gmail.com}
\affiliation{QICI Quantum Information and Computation Initiative, School of Computing and Data Science, The University of Hong Kong, Pokfulam Road, Hong Kong SAR, China}

\author{Fei Shi}
\email{shif26@mail.sysu.edu.cn}
\affiliation{Institute of Quantum Computing and Software, School of Computer Science and Engineering, Sun Yat-sen University, Guangzhou 510006, China}

\author{Giulio Chiribella}
\email{giulio@hku.hk}
\affiliation{QICI Quantum Information and Computation Initiative, School of Computing and Data Science, The University of Hong Kong, Pokfulam Road, Hong Kong SAR, China}

\author{Qi Zhao}
\email{zhaoqi@cs.hku.hk}
\affiliation{QICI Quantum Information and Computation Initiative, School of Computing and Data Science, The University of Hong Kong, Pokfulam Road, Hong Kong SAR, China}

\begin{abstract}
A fundamental problem in quantum physics is to establish whether a multiparticle quantum state can be uniquely determined from its local marginals. 
In theory, this problem has been addressed in the exact case where the marginals are perfectly known. 
In practice, however, experiments only have access to finite statistics and therefore can only determine the marginals of a quantum state up to an error. 
In this Letter, we prove that unique determinability universally survives such local imperfections: specifically, for every uniquely determined state, we show that deviations of local marginals propagate to global states strictly bounded by a power law with exponent $\alpha\in(0,1]$.
This result induces a classification of multipartite quantum states by their power-law exponents, with linear scaling $\alpha=1$ as the most favorable regime.
We derive a necessary and sufficient criterion for linear robustness and translate it into an executable semidefinite programming certification.
Applying our theory, we prove that stabilizer states are at least square-root robust and provide a complete robustness classification for the Dicke family.
Finally, we exploit these results to construct a scalable two-local genuine multipartite entanglement witness, demonstrating the viability of this framework for broad practical applications.
\end{abstract}
\maketitle

\emph{Introduction.---}
How much of a global quantum state is encoded in local marginals?
This interplay, formally known as the quantum marginal problem~\cite{coleman1963structure,klyachko2004quantum,klyachko2006quantum,eisert2008gaussian,schilling2015quantum}, traces back to Schrödinger's formulation of entanglement~\cite{schrodinger1935gegenwartige}.
It has long been recognized as a major challenge in quantum chemistry~\cite{national1995mathematical}, and exhibits a close connection to the generalized Pauli principle~\cite{altunbulak2008pauli}.
 Mathematically, determining whether there exists a global quantum state that fits a given set of marginals is hard and generally intractable~\cite{liu2007quantum}.
 However,   in most physical scenarios the marginals are guaranteed to come from an underlying global state, and the question is whether a global state is uniquely determined among all states (UDA) by its marginals~\cite{linden2002almost,linden2002parts,diosi2004three,walck2009only}. 
The UDA property is essential for reconstructing or certifying global properties from local measurements, including state tomography~\cite{xin2017quantum,yu2023learning,ona2025determining}, entanglement detection~\cite{walter2013entanglement,chen2014role,shi2025entanglement}, and self-testing~\cite{li2018self,aloy2021quantum}.

Early results indicated that almost all few-qubit pure states are UDA~\cite{linden2002almost}, and that this generic behavior extends to multipartite systems given proportionally scaled size of marginals~\cite{linden2002parts,jones2005parts,huang2018quantum}. 
Subsequent studies identified state families retaining the UDA property from much smaller marginals~\cite{parashar2009n,tabia2022entanglement,chen2026nonlocality}, and provided bounds on the number of local measurements~\cite{chen2013uniqueness,huang2018quantum,chen2026nonlocality}.
The UDA property was later connected to quantum many-body physics, showing that the ground states of local Hamiltonians, when unique, are UDA~\cite{chen2012ground,chen2012correlations,swingle2014reconstructing,huber2016characterizing}.
On the other hand, later research revealed that UDA states are not necessarily unique ground states of local Hamiltonians supported on the marginals' positions, thereby motivating a search for more systematic certification~\cite{karuvade2019uniquely}.

\begin{figure}[t]
    \centering
    \includegraphics[width=\linewidth]{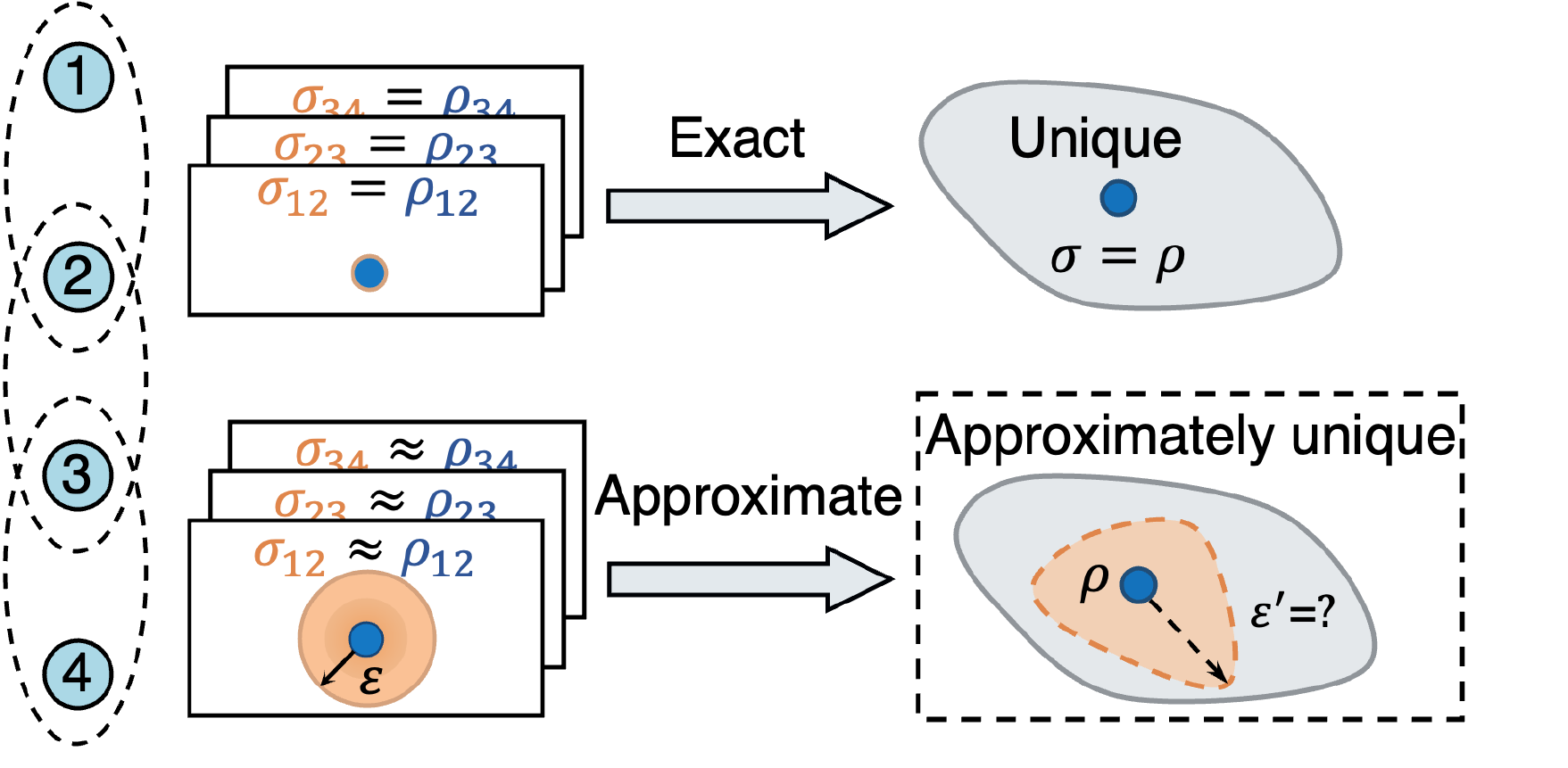}
    \caption{{\bf Robustness of UDA.} Top: for  UDA states, exact equality of the marginals on a subset of systems (here, all pairs of nearest neighboring qubits) implies exact equality of global states.    
    Bottom: for robust UDA states, approximate equality of the marginals implies approximate equality of the global states. 
    Then, the problem becomes to quantify the relation between the deviation of the marginals ($\varepsilon$) and the maximum deviation of the global states ($\varepsilon'$). }
    \label{fig:concept}
\end{figure}

The existing theory of UDA, however, assumes that the marginals are exactly known, whereas in real-world experiments the local marginals are only known up to inevitable errors.
This observation raises a fundamental question: is the UDA property robust under perturbations of the marginals?     In other words, if the marginals of a quantum state are close to the marginals of a UDA state, is the global state close to that UDA state?    
The answer to this question is far from obvious. 
First of all, satisfying approximate constraints need not generally ensure closeness to ideal behavior, as is known, for example, in the related problem of tackling approximate quantum Markov chain states~\cite{ibinson2008robustness,christandl2012entanglement}. 
A further difficulty comes from the lack of a general characterization of the set of UDA states, rendering it challenging to analyze the approximate version of the UDA property.   
Finally, even once the robustness of UDA is established, for practical usage
we desire quantitative bounds on how much the global state deviates from a given UDA state for a given size of local deviations, as illustrated in Figure~\ref{fig:concept}.

Here, we prove that the UDA property is universally robust under perturbations of the marginals:  for every UDA state, small deviations $\varepsilon$ in its marginals propagate to global deviations $\varepsilon'$ only with a power law $\varepsilon'  \propto  \varepsilon^{\alpha}$ for some $\alpha\in(0,1]$.
This result establishes a universal power law for UDA robustness that precludes ill-conditioned behavior and induces a classification of UDA states based on their power-law exponents.  The most favorable scaling arises in  UDA states with linear robustness,  for which the errors in the marginals propagate only linearly to the global state. To identify these states, we provide a necessary and sufficient criterion, as well as a certification method based on semidefinite programming.     
Complementing these general results, we analyze representative families of UDA states with different robustness exponents and highlight some of their applications. 
In particular, we prove that stabilizer states are at least square-root robust and completely classify the robustness of the Dicke state family. 
Based on the Dicke states' robustness, we further construct a scalable genuine multipartite entanglement witness relying solely on two-local measurements.  
Overall, these results establish the notion of UDA as a viable practical tool for characterizing multipartite entangled states.
\vspace{0.5em}

\emph{Framework for UDA robustness.---}   Let us start by reviewing the notion of UDA state. Let $\mathcal H$ be the finite-dimensional Hilbert space of a composite system with $n$ subsystems indexed by $\{1,\dots, n\}=: [n]$,   and let $\mathcal D(\mathcal H)$ be the set of density matrices (positive semidefinite, trace-1 matrices) on $\mathcal H$. 
Given a density matrix  $\rho\in\cD(\cH)$ and a subsystem $S\subseteq[n]$,  we denote the corresponding marginal (reduced density matrix)  by $\rho_S\coloneqq\Tr_{\bar S}(\rho)$, where $\Tr_{\bar S}$ denotes the partial trace over  $\bar S=[n]\setminus S$.
 For a collection of subsystems $\cS\coloneqq\{S_k\}_{k=1}^M$ ($S_k \subseteq [n]$), the \emph{compatibility set} of the state $\rho$ is defined as:
\begin{equation}
    \cC_\cS(\rho)\coloneqq\big\{\sigma\in\cD(\cH) ~\big|~ \sigma_S= \rho_{S}, \ \forall  S\in \cS \big\}.
\end{equation}
This set contains all global states sharing $\rho$'s marginals on $\cS$.
The state $\rho$ is UDA with respect to $\cS$ if $\cC_{\cS}(\rho)=\{\rho\}$.

We now recast the notion of UDA state into an analytical expression that facilitates the transition to approximate UDA.   
Let   $\mathrm{Herm} (\mathcal H)$ be the set of Hermitian operators on $\mathcal H$, and let $\mathcal{V} \coloneqq \{X \in \mathrm{Herm}(\mathcal{H})\mid\Tr (X) = 0\}$ be  the space of traceless Hermitian operators.
For $X\in\cV$, the \emph{marginal map} $\mathcal{M}_{\mathcal S}(X)\coloneqq(X_S)_{S\in\mathcal S}$ collects its marginals on $\mathcal S$.
Its kernel, $\mathcal{W}_{\mathcal S}\coloneqq\{X\in\cV\mid\mathcal{M}_{\mathcal S}(X)=\mathbf{0}\}$, consists of operators whose marginals on $\mathcal S$ all vanish, where $\mathbf{0}$ denotes the corresponding list of zero matrices.
Two states $\rho$ and $\sigma$ have the same marginals if and only if $\sigma-\rho\in\mathcal{W}_{\mathcal S}$.
Defining the \emph{state difference set} $\mathcal{D}_0(\rho)\coloneqq\{\sigma-\rho\mid\sigma\in\mathcal{D}(\mathcal H)\}$, the UDA condition is therefore
\begin{equation}\label{eq:UDA}
    \mathcal{D}_0(\rho) \cap \mathcal{W}_\mathcal{S} = \{\mathbf{0}\}\, .
\end{equation}

Let us now move to the approximate case.  
To quantify the deviations between local marginals of two states $\rho$ and $\sigma$, we sum the trace-norm distances between corresponding marginals.
With $\|A\|_1\coloneqq\Tr\sqrt{A^\dag A}$ denoting the trace norm, we define the \emph{marginal norm}
\begin{gather}
    \bigl\|\cM_\cS(X)\bigr\|_{\mathcal S}
    \coloneqq
    \sum_{S  \in  \mathcal S} \|X_{S}\|_1 \, ,
\end{gather}
for a generic Hermitian operator $X\in \cV$.  
The main question of UDA robustness then turns to:
how large can the trace distance $\| \sigma-\rho\|_1$ be under the constraint that the marginal norm   $\bigl\|\cM_\cS(\sigma-\rho)\bigr\|_{\mathcal S}$  is less than $\varepsilon$?

\vspace{0.5em}

\emph{Universal robustness.---} Our first result shows that, in finite dimensions, the UDA property alone guarantees power-law robustness for pure and mixed states.

\begin{theorem}[Universal robustness of UDA]\label{thm:main}
Let $\rho\in\cD(\cH)$ be UDA with respect to $\cS$, where $\cH$ is finite-dimensional. 
Then there exist constants $C>0$, $0<\alpha\le1$, and $\varepsilon_0>0$ such that every $\sigma \in \mathcal{D}(\mathcal{H})$ with
\begin{gather}\label{eq:margin}
    \varepsilon=\norm{\mathcal{M}_{\mathcal{S}}(\sigma-\rho) }_{\mathcal{S}}
\le
\varepsilon_0,
\end{gather}
satisfies
\begin{gather}
    \norm{\sigma - \rho}_1
\le
C \, \varepsilon^{\alpha}.
\end{gather}
\end{theorem}
\noindent Theorem \ref{thm:main} shows that UDA states  satisfy a favorable robustness property: if the trace distance between  the marginals of a UDA state and the marginals of another state are  small, then the trace distance between the two global states is small too, and the dependence  of the global trace distance on the marginal trace distance has a   polynomial scaling.   
This conclusion is universal: it does not require any  further conditions such as symmetry or a unique-ground-state representation and holds regardless of classical simulability (see Supplementary Material Sec. IV.D).

To establish the above result, we consider an arbitrary global deviation $X\in \cV$.
The marginal norm of $X$ controls its distance
to the invisible subspace $\cW_{\cS}$,
\begin{gather}
    \dist(X,\cW_\cS)\coloneqq\inf_{Y\in\cW_\cS}\norm{X-Y}_1.
\end{gather} 
To view this, we pass to the quotient space $\mathcal{V}/\cW_\cS$ to remove the direction invisible to marginals.
With $[X]=X+\cW_\cS$, the induced map $\widetilde{\mathcal{M}}_{\mathcal{S}}([X]) \coloneqq \mathcal{M}_{\mathcal{S}}(X)$ is bijective (thus invertible) from $\mathcal{V}/\mathcal{W}_{\mathcal{S}}$ onto $\mathrm{Im}(\mathcal{M}_{\mathcal{S}})$.
Equipped with the quotient norm $\norm{[X]}_{\mathcal{V}/\mathcal{W}_{\mathcal{S}}}\coloneqq\operatorname{dist}(X,\mathcal{W}_{\mathcal{S}})$, the finite-dimensional quotient space guarantees a bounded inverse of $\widetilde{\mathcal{M}}_{\mathcal{S}}([X])$~\cite{kreyszig1978introductory}, so
\begin{align}\label{eq:bounded_inverse}
    \dist(X,\cW_\cS)=&\norm{[X]}_{\cV/\cW_\cS}=\norm{\widetilde{\cM}_\cS^{-1}\circ\widetilde{\mathcal{M}}_{\mathcal{S}}([X])}_{\cV/\cW_\cS}\notag\\
    \leq&\norm{\widetilde{\cM}_\cS^{-1}}_{\mathrm{op}}\cdot\norm{\cM_\cS(X)}_\cS,\ \forall\,X\in\cV. 
\end{align}

To turn the distance into a small global size, we restrict to physical differences $\delta\in\cD_0(\rho)$ from the target state $\rho$.
The trace and positivity constraints defining $\mathcal{D}_0(\rho)$ can be represented by finitely many polynomial inequalities, making it a semialgebraic set~\cite{bochnak1998real}.
At any isolated intersection, these polynomial constraints allow tangent contact, but rule out infinite flatness~\cite{TaL1995,Solerno1991}. 
Since the UDA condition gives $\cD_0(\rho)\cap\cW_\cS=\{\mathbf 0\}$, 
this geometric property ensures that 
\begin{gather}\label{eq:visible_part}
    \|\delta\|_1 \le K \dist(\delta,\cW_\cS)^{\alpha},\ \forall\,\delta\in\cD_0(\rho),
\end{gather}
where $K>0$ and $\alpha\in(0,1]$.
Intuitively, this bound asserts that no physical global deviation can remain arbitrarily invisible to the selected marginals while still being large.
Since $\cD_0(\rho)\subset\cV$, combining this bound with the linear bound in Eq.~\eqref{eq:bounded_inverse} yields Theorem~\ref{thm:main}.
The resulting power-law guarantee holds for every finite-dimensional UDA instance, while the specific scaling may depend on the target state, local dimension, system size, and marginal sizes.
We defer the complete proof to the Supplementary Material (SM.) Sec.~II

The exponent $\alpha$ characterizes how local deviations in the marginals propagate to deviations of the global quantum state. 
One can thus classify UDA states related to \emph{$\alpha$-robustness}~\footnote{Note that an $\alpha$-robust state is also trivially $\alpha'$-robust for any $\alpha' \le \alpha$. However, it is generally interesting to consider the maximum $\alpha_\star$ for a UDA state.}. 
The largest exponent corresponds to $\alpha=1$, a condition named \emph{linear robustness}.  
In the following section, we characterize when this optimal robustness holds for a UDA state. 

\vspace{0.5em}

\emph{Linear robustness.---}
While every UDA state obeys a power-law robustness bound, the relation in Eq.~\eqref{eq:bounded_inverse} is already linear. 
Hence, any loss from the linear robustness originates in Eq.~\eqref{eq:visible_part} and is governed by how rapidly physically valid state differences separate from the invisible subspace $\cW_\cS$ near their isolated intersection at $\mathbf{0}$.
We now derive a necessary and sufficient criterion for linear robustness and turn it into an executable certification procedure.
Full proofs of all results in this section are provided in SM. Sec.~III.

For a UDA state $\rho$ with respect to $\cS$, we define the \emph{tangent cone} of the corresponding state-difference set $\cD_0(\rho)$ by
\begin{equation}
\label{eq:tangent-def}
\cK_{\cD_0(\rho)}(\mathbf{0}) \coloneqq \left\{ X \in \cV \ \middle| \ X = \lim_{k\to\infty} \frac{\delta_k}{t_k},\  t_k \to 0^+ \right\},
\end{equation}
where $\{\delta_k \in \cD_0(\rho)\}_k$ denotes a sequence of differences between state $\rho$ and states $\sigma_k$ that approach $\rho$ in the limit of $t_k$ going to $0$, and $\{t_k\}_k$ denotes a sequence of real numbers approaching $0$.  
This cone collects all valid \emph{directions} in which a state $\sigma$ can deviate from $\rho$. We now show that  linear robustness can be completely characterized in terms of these directions:
\begin{theorem}[Linear robustness criterion]\label{thm:linear}
     A UDA state $\rho$ with respect to the marginal set $\cS$ is linearly robust if and only if $\mathcal{K}_{\mathcal{D}_0(\rho)}(\mathbf{0}) \cap \mathcal{W}_{\mathcal{S}} = \{\mathbf{0}\}$.
\end{theorem}
\vspace{-0.3em}
This theorem provides a complete criterion to determine whether a state has linear robustness. 
Specifically, one needs  to check  whether there exists a nonzero direction $X\in \mathcal{K}_{\mathcal{D}_0(\rho)}(\mathbf{0})$, in which states can deviate from $\rho$,  such that  the marginals of $X$ on the subsystems in  $\mathcal{S}$ are zero ($X\in \mathcal{W}_{\mathcal{S}}$).  If there exists such an $X$, then the state $\rho$ fails to have linear robustness. If not, $\rho$ has linear robustness.

\begin{figure}
    \centering
    \includegraphics[width=\linewidth]{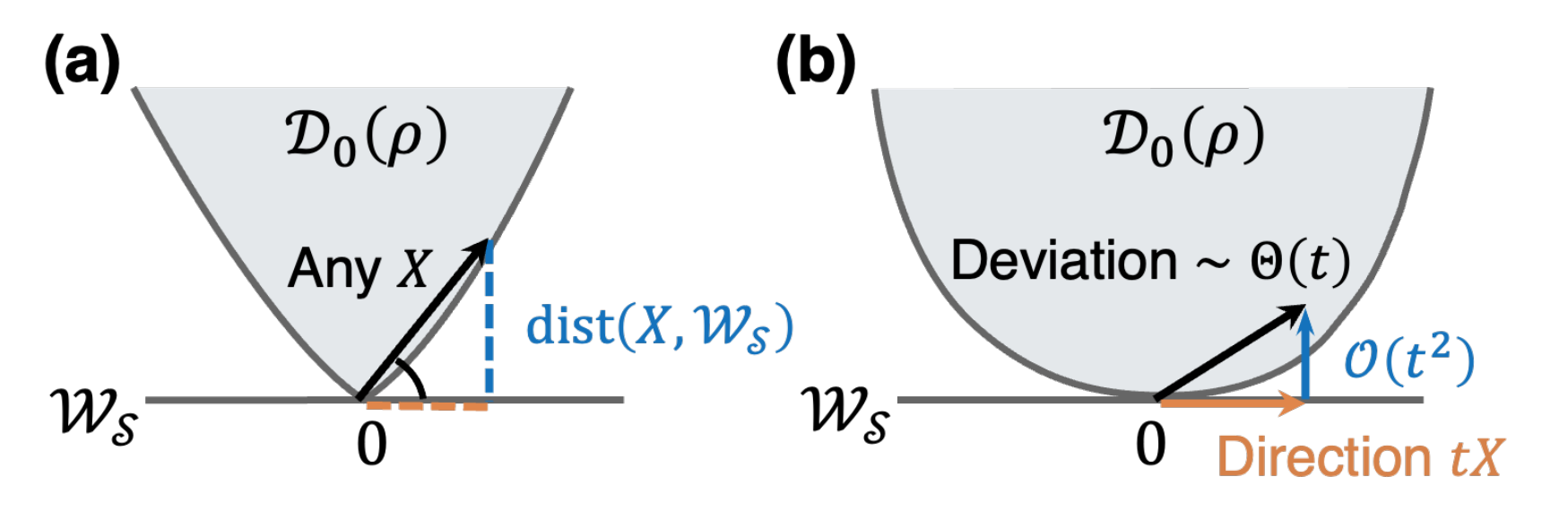}
    \caption{\textbf{Geometry near the intersection $\mathbf{0}$.} 
   Semialgebraicity ensures that the intersection is either transverse or polynomially tangential. \textbf{(a)} Transverse intersection. 
    All valid directions diverge from $\cW_\cS$ at a strict angle.
    Thus, their distance from $\cW_\cS$ (blue dashed line) are proportional to their own sizes, ensuring linear robustness. 
    \textbf{(b)} Tangential intersection. 
    The boundary smoothly aligns with the subspace $\cW_\cS$ in \textbf{at least one} direction (orange arrow).
    A second-order correction produces a valid state path, which exhibits a first-order size but only the second-order distance from $\cW_\cS$.}
    \label{fig:intersection}
\end{figure}

Some insight into the proof of Theorem \ref{thm:linear} comes from considering two mutually exclusive geometries, namely the \emph{transverse} and \emph{tangential} geometry, illustrated in Figure~\ref{fig:intersection}. 
In the transverse case, where $\mathcal{K}_{\mathcal{D}_0(\rho)}(\mathbf{0})\cap\mathcal{W}_\mathcal{S}=\{\mathbf{0}\}$, every nonzero direction leaves the origin at a positive angle to the subspace $\cW_\cS$ and thus has a distance from $\cW_\cS$ proportional to its total size, as in Figure~\ref{fig:intersection}(a).
Concretely, we consider normalized directions $S_\cK\coloneqq\{X\in \cK_{\cD_0(\rho)}(\mathbf{0}):\, \norm{X}_1=1\}$, which form a compact set.
The geometry ensures a positive minimum distance $\kappa\coloneqq \min_{X\in S_\cK}\dist(X,\cW_\cS)>0$.
For any direction $X\in\cK_{\cD_0(\rho)}(\mathbf{0})$, we have
\begin{equation}
\label{eq:transverse-cone}
\dist(X,\mathcal W_{\mathcal S})=\norm{X}_1\,\operatorname{dist}\!\left(\frac{X}{\norm{X}_1},\cW_\cS\right)\ge \kappa\,\|X\|_1,
\end{equation}
where $\kappa$ is determined by the angle of the transverse intersection.
Since $\mathcal{D}_0(\rho)\subseteq\cK_{\cD_0(\rho)}(\mathbf{0})$, Eq.~\eqref{eq:transverse-cone} applies to every state difference $\delta\in\mathcal{D}_0(\rho)$, offering the linear version of Eq.~\eqref{eq:visible_part}.
Combining with Eq.~\eqref{eq:bounded_inverse} ensures the linear robustness of $\rho$.

Conversely, suppose $\cK_{\cD_0(\rho)}(\mathbf{0})\cap\cW_\cS\neq\{\mathbf{0}\}$.
In this tangential case, there exists a nonzero first-order direction $X$ that is invisible to local marginals, as shown in Figure~\ref{fig:intersection}(b). 
For any such direction $X$, we can construct a positive semidefinite matrix with small $t$:
\begin{gather}\label{eq:construct_quadra}
    \widetilde\rho(t) \coloneqq\rho+tX+t^2 R,
\end{gather}
where the correction $R$ maintains positive semidefiniteness (via Schur complements~\cite{zhang2006schur}) and is specified explicitly in SM. Sec. III.
Normalizing $\widetilde\rho(t)$ gives a valid state $\sigma(t)$ and only changes the second- and higher-order terms. 
The global state therefore deviates at first order,
$\norm{\sigma(t)-\rho}_1=t\norm{X}_1+\order{t^2}$.
By contrast, since
$X\in\mathcal W_{\mathcal S}$, the marginal map omits
the first-order term such that $\norm{\cM_\cS(\sigma(t)-\rho)}_\cS=\order{t^2}$.
Thus, $\{\sigma(t)\}$ demonstrates a trace of deviated states with square-root propagations of local deviations to the global side, thereby precluding the linear robustness of $\rho$.

This construction of $\{\sigma(t)\}$ also shows that the absence of strong sensitivity regarding the deviation direction necessarily excludes the entire intermediate range of robustness exponents.
With the equivalence between the linear robustness and the sensitivity, this leads to
 \begin{corollary}[Robustness Gap]\label{coro:gap}
    If a UDA state $\rho$ is not linearly robust, $\rho$ is not $\alpha$-robust for any $1/2<\alpha<1$.
 \end{corollary}
\vspace{-0.3em}
The abstract tangent cone in Eq.~\eqref{eq:tangent-def} simplifies algebraically. 
A valid deviated state must maintain unit trace and positive semidefiniteness.
Let $P_0$ denote the projector onto $\ker\rho$.
The tangent cone has the explicit form
\begin{equation}
\label{eq:K-explicit-sec}
\cK_{\cD_0(\rho)}(0)=\{X\in\cV\ |\   P_0XP_0\succeq 0\}.
\end{equation}

This transforms our geometric criterion into an executable search for a nonzero invisible direction $X\in\cK_{\cD_0(\rho)}(\mathbf{0})\cap\cW_\cS$.
First, we solve the feasibility of a linear program,
\begin{equation}
\label{eq:Pl-sec}
X\in \cV\ \text{s.t.}\ \mathcal M_{\mathcal S}(X)=\mathbf{0},\ P_0 X P_0= \mathbf{0}.
\end{equation}
This checks for zero-marginal directions that do not affect the null space. 
If this admits only $X=\mathbf{0}$ and $P_0 \neq \mathbf{0}$, we test the semidefinite feasibility program:
\begin{equation}\label{eq:PS-sec}
X\in \cV\, \text{s.t.}\,\mathcal M_{\mathcal S}(X)=\mathbf{0},\, P_0 X P_0\succeq 0,\, \Tr(P_0 X P_0)=1.
\end{equation}
This identifies zero-marginal directions with positive null-space components. 
A nonzero solution of Eq.~\eqref{eq:Pl-sec}, or feasibility of
Eq.~\eqref{eq:PS-sec}, rules out linear robustness.
If neither occurs, $\rho$ is linearly robust.

Equation~\eqref{eq:Pl-sec} also yields a simple information-counting
constraint.
For a rank-$r$ linearly robust state, $\mathcal{M}_\mathcal{S}$ must be injective on the traceless Hermitian space satisfying $P_0XP_0=0$.
This space has real dimension
$(r^2-1)+2r(d-r)$, counting variations within the support and coherences between the support and
kernel.
In contrast, $M$ marginals of at most
$s$ qubits provide at most $M4^s$ real coordinates.
The former cannot exceed the latter, leading to:
\begin{corollary}[Marginal size condition]
\label{cor:lower-bound-s}
Let $\rho$ be an $n$-qubit state of rank $r$ that is UDA
with respect to the marginal set $\cS=\{S_k\}_{k=1}^M$,
and set $d=2^n$.
If $\rho$ is linearly robust,
the largest marginal size $s\coloneqq\max_k |S_k|$ must satisfy
\begin{equation}
\label{eq:s-lower}
s\, \ge\, \frac{\log_2\!\bigl((r^2-1)+2r(d-r)\bigr)-\log_2 M}{2}.
\end{equation}
\end{corollary}
\noindent In other words, linear robustness requires sufficiently informative marginal data.
For pure $n$-qubit states, it implies the simpler bound $2s+\log_2M\geq n$.
One must either measure a set of proportionally large subsystems, or collect exponentially many smaller marginals to achieve linear robustness.
An interesting question is whether this tradeoff holds for arbitrary UDA states, beyond the linear robustness regime: that is, whether access to larger marginals could always be replaced by use of a greater number of smaller marginals.
Corollary \ref{cor:lower-bound-s} can also be used as guidance to design sets of  marginals that achieve  linear robustness. In the following, we will illustrate this application in the case of Dicke states.

\vspace{0.5em}

\emph{Case studies.---}
Preceding theoretical results endow UDA with rigorous robustness under imperfect local data. 
We now revisit representative UDA families to determine the explicit robustness and showcase a corresponding practical application.
The detailed proofs and calculations are deferred to the End Matter and SM. Sec.~IV.

A common route to UDA arises from the local Hamiltonian.
If a pure state $\rho$ is the unique ground state of a Hamiltonian $H=\sum_{i=1}^m H_i$ with local terms $\{H_i\}$, any state sharing identical marginals has the same energy and must coincide with it.
The spectral gap further converts the energy difference from marginal errors into the global trace distance~\cite{cramer2010efficient,swingle2014reconstructing,yu2023learning}.
We rephrase this mechanism as a bridging proposition establishing square-root robustness.
\begin{proposition}\label{prop:UniqueGround}
    Suppose the pure state $\rho$ is the unique ground state of an $n$-qubit Hamiltonian $H=\sum_{i=1}^m H_i$ with spectral gap $\Delta>0$ and local supports $\cS\coloneqq\{S_1,\cdots,S_m\subseteq[n]\}$.
    Then $\rho$ is a UDA state with respect to marginal supports $\cS$ such that for any $\sigma\in\cD(\cH)$
    \begin{gather}\label{eq:quadra}
        \norm{\sigma-\rho}_1\leq 2\sqrt{\frac{\omega_{\max}}{\Delta}}\norm{\cM_\cS(\sigma-\rho)}_\cS^{1/2},
    \end{gather}
    where $\omega_{\max} \coloneqq\max_{i=1}^m(\lambda_{\max}(H_i)-\lambda_{\min}(H_i))/2$.
\end{proposition}

This proposition establishes a connection between spectral properties of unique ground state Hamiltonians and metric properties relating local trace distance to global trace distance. Using this connection,     we now  examine the case of  pure stabilizer states~\cite{PhysRevA.92.012305,shi2025entanglement}.
Suppose $\rho\coloneqq\ket{\psi}\!\bra{\psi}$ is an $n$-qubit stabilizer state regarding a maximal stabilizer group $\mathbb{G}$ and independent generators $g_1,\cdots,g_n\in\mathbb{G}$.
The parent Hamiltonian $H=\sum_{i=1}^n \Pi_i$, where $\Pi_i=(I-g_i)/2\succeq 0$, has $\rho$ as its unique ground state.
Every state orthogonal to $\ket{\psi}$ has energy at least $1$.
Thus, Proposition~\ref{prop:UniqueGround} offers a baseline of square-root robustness for $\rho$ with respect to the generators' supports.

The linear robustness certification also simplifies substantially here.
The certification checks for the feasibility of programs over traceless Hermitian operators $X\in\cV$ such that $\Tr(HX)=0$ and $P_0XP_0\succeq \mathbf{0}$.
Moreover, we know that the parent Hamiltonian 
satisfies $H\succeq I-\rho=P_0$. This forces $P_0XP_0=0$, and the linear
certification reduces to only checking the linear program in Eq.~\eqref{eq:Pl-sec}.

To illustrate this simplified certification, we analyze one-dimensional cluster and ring states.
Because these states are foundational resources for measurement-based quantum computation, their UDA robustness suggests a practical local verification guarantee.
These graph states are defined by local stabilizer generators, $g_i = Z_{i-1}X_i Z_{i+1}$, with the two end generators $g_1=X_1Z_2$ and $g_n=Z_{n-1}X_n$ for cluster states, and periodic boundary conditions for ring states. 
Since Proposition~\ref{prop:UniqueGround} guarantees at least square-root robustness, the failure of the certification directly leads to the exact square-root robustness based on Corollary~\ref{coro:gap}.
We thus numerically execute the certification and summarize the robustness in Table~\ref{tab:graph-states}.

\begin{table}[t]
\centering
\renewcommand{\arraystretch}{1.22}
\setlength{\tabcolsep}{6pt}
\begin{tabular*}{\columnwidth}{@{\extracolsep{\fill}}lcccc}
\hline\hline
State family & $n=4$ & $n=5$ & $n=6$ & $n=7$ \\
\hline
\makecell[l]{Cluster state} & $\alpha_\star=1$ & $\alpha_\star=1/2$ & $\alpha_\star=1/2$ & $\alpha_\star=1/2$ \\
\makecell[l]{Ring state}     & $\alpha_\star=1$ & $\alpha_\star=1$ & $\alpha_\star=1/2$ & $\alpha_\star=1/2$ \\
\hline\hline
\end{tabular*}
\caption{\label{tab:graph-states}Exact (maximum) robustness exponent $\alpha_\star$ characterization for $n$-qubit graph states. 
Cluster states are tested with 3-local chain marginals plus two end marginals, and ring states with 3-local ring marginals that respect periodic boundary conditions.}
\end{table}

We next turn to the highly symmetric Dicke family.
For any $n$-qubit system ($n\geq2$), we denote the Dicke state with Hamming weight $1\leq k\leq n-1$ as
\begin{gather}\label{eq:dicke}
    \ket{D(n,k)}\coloneqq\frac{1}{\sqrt{\binom{n}{k}}}\sum_{\substack{\mathbf{x}\in\{0,1\}^n\\\textsf{wt}(\mathbf{x})=k}}\ket{\mathbf{x}}.
\end{gather}
We adopt a parent Hamiltonian for $\ket{D(n,k)}$~\cite{chen2012correlations}:
\begin{gather}\label{eq:H_nk}
H_{n,k}\coloneqq\sum_{1\le i<j\le n}\Pi^-_{ij}+\frac{1}{4}\left[\sum_{i=1}^nZ_i-(n-2k)I\right]^2,
\end{gather}
where $\Pi^-_{ij}\coloneqq(I-\text{SWAP}_{ij})/2$.
The first term quantifies the antisymmetry between every pair of qubits, and the second term accounts for the Hamming weights of states. 
Their combination makes $\ket{D(n,k)}$ the unique ground state.
Since $H_{n,k}$ is a sum of two-local terms, Dicke states are at least square-root robust UDA states with respect to the full collection of two-local marginals.
Optimizing the relative weights in this construction yields the robustness coefficient,
\begin{gather}
    2\sqrt{\frac{\omega_{n,k}}{\Delta_{n,k}}}=\sqrt{2+\frac{2\abs{n-2k}}{n-1}}.
\end{gather}

Our robustness framework also admits precise analysis of when the Dicke-state family can maintain linear robustness.
\begin{proposition}[Dicke marginal-size threshold]
\label{prop:dicke-threshold-main}
Let $\rho_{n,k}=\ket{D(n,k)}\!\bra{D(n,k)}$, where $n\geq3$
and $1\leq k<n$, and define
\begin{gather}
s_\star(n,k)
\coloneqq
\max\{k,n-k\}+\mathbf1_{\{2k=n\}}.
\label{eq:dicke-threshold-main}
\end{gather}
$\ket{D(n,k)}$ is linearly robust with respect to all $s_\star(n,k)$-body marginals.
Moreover, no marginal collection whose largest support is smaller
than $s_\star(n,k)$ can yield linear robustness.
\end{proposition}
\noindent This threshold quantifies the cost of maintaining
$\alpha_\star=1$ across system sizes.
The proof is given in SM. Sec.~IV.B.
In particular, setting $s=2$ recovers a classification for the ordinary two-body case:
all nontrivial Dicke states with $n\geq3$ are exactly square-root
robust with respect to all two-body marginals, except for the
linearly robust $\ket{D(3,1)}$ and $\ket{D(3,2)}$.

As a practical application, the two-body robustness bound yields a local and scalable genuine multipartite entanglement (GME) witness. 
Standard projective GME witnesses require measuring the global state fidelity, which typically necessitates cumbersome global measurements.
By applying our robustness bounds, we bypass this bottleneck, lower-bounding the global fidelity using only deviations of two-local marginals.
This results in an experimentally friendly GME certification protocol that relies solely on two-body measurements. 
We detail this in Appendix B.

Our framework also applies to general UDA states beyond the unique ground states.
Consider the six-qubit state,
\begin{gather}
\ket{\Psi_6}
=
\frac{1}{\sqrt2}
\left(
\ket{0^6}
+
\frac13
\sum_{\{i,j\}\in E_{\mathrm{far}}}\ket{ij}
\right),
\label{eq:nonugs-state}
\end{gather}
where $E_{\mathrm{far}}$ is the set of nine non-neighboring
pairs on the six-cycle, and $\ket{ij}$ denotes the six-qubit
computational-basis state with excitations only at sites $i,j$.
This state is UDA from all two-body marginals but is not the
unique ground state of any two-local Hamiltonian~\cite{karuvade2019uniquely}.
Theorem~\ref{thm:main} guarantees its robustness, and our further
analysis yields a bound with exponent $\alpha=1/4$ and an explicit
constant (see SM. Sec.~IV.C).

\vspace{0.5em}

\emph{Conclusion and outlook.---}
We proved that the set of multipartite quantum states uniquely determined by their marginals is robust against perturbations:  any global state's deviation from a UDA target is strictly bounded by a power law of their local marginals' deviations. 
The most favorable case corresponds to a linear relation between the distance of the marginals and the distance of the global states. 
To identify states with this property,  we provided a necessary and sufficient criterion, as well as a set of executable tests.  

This robustness theory promotes UDA from an interesting theoretical concept to a practically relevant tool, providing a theoretical foundation for applications to state tomography, entanglement certification, and, more generally, to every task involving the estimation of global properties from local measurements.

\begin{acknowledgments}
  \emph{Acknowledgments.---}
This work was supported by the Chinese Ministry of Science and Technology (MOST) through grant 2023ZD0300600.  
F.S. acknowledges funding from the National Natural Science Foundation of China (Grant No. 62601907 and No. 12571493), and the 2026 Basic Start-up Fund of Sun Yat-sen University (Grant No. 67000-12266020). 
G.C. acknowledges Hong Kong Research Grant Council (RGC) through grants  SRFS2021-7S02  and R7035-21F, and the State Key Laboratory of Quantum Information Technologies and Materials.
Q.Z. acknowledges funding from Quantum Science and Technology-National Science and Technology Major Project 2024ZD0301900, National Natural Science Foundation of China (NSFC) via Project No. 12347104 and No.~12305030, Hong Kong Research Grant Council (RGC) via No.17310926, N\_HKU718/23, and R6010-23, the Guangdong Provincial Quantum Science Strategic Initiative (Grants No.GDZX2503001).

  \emph{Code \& data availability.---}All code and data for linear certification is available at~\cite{Wenjun26} 
\end{acknowledgments}

\bibliographystyle{apsrev4-2}
\makeatletter\auto@bib@empty\makeatother

\newpage

\appendix

\begin{center}
    \textbf{End Matter}
\end{center}

\emph{Appendix A: Unique Ground State Mechanism.---}Here we provide the formal proof that a pure state uniquely determined as the gapped ground state of a local Hamiltonian universally exhibits at least square-root robustness.
\begin{proof}[Proof of Proposition~\ref{prop:UniqueGround}]
   The pure state $\rho$ is the unique ground-state projector of $H$. 
   The spectral decomposition gives the operator inequality $H - E_0 I \succeq \Delta(I - \rho)$, where $E_0$ is the ground state energy. 
   Taking the expectation with respect to an arbitrary state $\sigma \in \cD(\cH)$ gets:
   \begin{gather}\label{eq:energy_gap}
       \Tr((H - E_0 I)\sigma) \ge \Delta(1 - \Tr(\rho\sigma)).
   \end{gather}
    Since $H=\sum_{i=1}^m H_i$ and $E_0 = \Tr(H\rho)$, the energy difference can be expanded locally as $\sum_{i=1}^m \Tr(H_i(\sigma_{S_i} - \rho_{S_i}))$.
    Because the difference $\delta_i\coloneqq\sigma_{S_i} - \rho_{S_i}$ is traceless, replacing $H_i$ by $H_i-c_iI$ does not change the value for arbitrary constant $c_i$.
    Applying H\"{o}lder's inequality, we get $|\Tr(H_i \delta_i)| \le \|H_i - c_i I_{S_i}\|_\infty \|\delta_i\|_1$, where $\|\cdot\|_\infty$ denotes the operator norm (largest singular value).
    Optimizing over $c_i$ makes the smallest first term equal to the spectrum range $\omega_i$ of $H_i$.
    Picking the maximum of $\omega_i$ and summing these bounds over all supports $\cS$ provides an upper bound:
    \begin{gather}\label{eq:energy_marginal}
        \Tr((H - E_0 I)\sigma)\leq \omega_{\max}\sum_{i=1}^m\norm{\sigma_{S_i}-\rho_{S_i}}_1.
    \end{gather}
    Chaining Eqs.~\eqref{eq:energy_gap} and~\eqref{eq:energy_marginal} immediately proves the infidelity bound:
    \begin{gather}\label{eq:infidelity}
        1 - \Tr(\rho\sigma) \le \frac{\omega_{\max}\sum_{i=1}^m\norm{\sigma_{S_i}-\rho_{S_i}}_1}{\Delta}.
    \end{gather}
    To translate this infidelity into the global trace distance, we invoke the Fuchs–van de Graaf inequalities~\cite{fuchs2002cryptographic}.
    For any pure state $\rho$, the trace distance is bounded by the fidelity: $\|\sigma - \rho\|_1 \le 2\sqrt{1 - \Tr(\rho\sigma)}$.
    Substituting the established infidelity bound completes the proof.
\end{proof}

\emph{Appendix B: Two-Local GME Witness.---}
A state $\rho$ is genuinely multipartite entangled if it cannot be decomposed as a convex combination of biseparable states. 
This convex hull is denoted by $\text{Bisep}$.
Inspired by this definition, a standard projective GME witness based on a pure state $\ket{\psi}$ takes the form~\cite{guhne2009entanglement}:
\begin{gather}
    W_\psi \coloneqq \beta(\rho_\psi) I - \rho_\psi,
\end{gather}
where $\rho_\psi = \ket{\psi}\!\bra{\psi}$, and $\beta(\rho_\psi) \coloneqq \sup_{\theta \in \text{Bisep}} \Tr(\rho_\psi \theta)$ is the maximum overlap with any biseparable state.
For any test state $\sigma \in \cD(\cH)$, a strictly negative expectation value $\Tr(W_\psi \sigma) < 0$ definitively certifies that $\sigma$ contains genuine multipartite entanglement.
For the Dicke state family $\rho_{n,k} = \ket{D(n,k)}\!\bra{D(n,k)}$ with $n\geq3$, the maximum biseparable overlaps are analytically known as~\cite{bergmann2013entanglement}:
\begin{gather}
    \beta(D(n,k)) = 
    \begin{cases}
\frac{n-k}{n} & \text{if } k < n/2 \\
\frac{n}{2(n-1)} & \text{if } k = n/2 \\
\frac{k}{n} & \text{if } k > n/2
\end{cases}.
\end{gather}

Directly evaluating $\Tr(W_\psi \sigma)$ typically requires measuring the global fidelity $\Tr(\rho_\psi \sigma)$, which involves global measurements. 
Fortunately, because Dicke states are UDA with respect to their two-local marginals, we can exploit the robustness to systematically lower-bound the global fidelity using exclusively two-body measurements.

Recall from Eq.~\eqref{eq:infidelity}, the deviations between marginals from a state and the target Dicke state strictly bound the global infidelity.
According to the calculation in SM. Sec.~IV, we established that the spectral gap is $\Delta_{n,k} = 1$ and the maximum local spectral spread is $\omega_{n,k} = \frac{1}{4}\left(2 + \frac{2|n-2k|}{n-1}\right)$.
Substituting this infidelity bound and coefficients into the global witness gives:
\begin{gather}
    \Tr(W_{n,k} \sigma) \le \beta(D(n,k)) - 1 + \omega_{n,k} \sum_{i=1}^m\norm{\sigma_{S_i}-\rho_{S_i}}_1.
\end{gather}
\noindent Consequently, we can strictly certify the GME nature of any state $\sigma$ whenever the measured two-local discrepancy satisfies the fully observable condition:
\begin{gather}
    \sum_{i=1}^m\norm{\sigma_{S_i}-\rho_{S_i}}_1<\frac{1 - \beta(D(n,k))}{\omega_{n,k}}.
\end{gather}
This protocol provides a highly practical application of UDA robustness bounds. 
By reconstructing only two-qubit reduced density matrices, we can rigorously confirm the presence of global GME in Dicke state preparations, which offers a scalable method for this problem compared to previous studies~\cite{shi2025entanglement,wu2025experimental}.

\onecolumngrid
\clearpage
\renewcommand*{\HyperDestNameFilter}[1]{supp-#1}
\hypersetup{urlcolor=black}
\renewcommand{\encodingdefault}{OT1}
\renewcommand{\rmdefault}{cmr}
\renewcommand{\bfdefault}{bx}
\renewcommand{\sfdefault}{cmss}
\renewcommand{\ttdefault}{cmtt}
\normalfont\normalsize
\makeatletter
\let\@hangfrom@section\mainSectionFormat
\let\@sectioncntformat\mainSectionCounterFormat
\renewcommand{\addcontentsline}[3]{}
\setcounter{secnumdepth}{3}
\@removefromreset{equation}{section}
\let\theequation@prefix\@empty
\let\@FMN@list\@empty
\patchcmd{\endthebibliography}{\label{LastBibItem}}{\label{sm:LastBibItem}}{}{\PackageError{uda-merge}{Cannot separate bibliography labels}{Check the REVTeX bibliography definition.}}
\makeatother
\renewcommand{\thesection}{\Roman{section}}
\renewcommand{\thesubsection}{\Alph{subsection}}
\renewcommand{\thesubsubsection}{\arabic{subsubsection}}
\renewcommand{\theequation}{\arabic{equation}}
\setcounter{section}{0}
\setcounter{subsection}{0}
\setcounter{subsubsection}{0}
\setcounter{equation}{0}
\setcounter{theorem}{0}
\setcounter{proposition}{0}
\setcounter{lemma}{0}
\setcounter{corollary}{0}
\setcounter{definition}{0}
\setcounter{algocf}{0}
\setcounter{figure}{0}
\setcounter{table}{0}
\begin{center}
\large\bfseries Supplementary Materials:
\end{center}
\vspace{9pt}
\section{Semialgebraic sets and functions}
In the main text, establishing the universal robustness of UDA states relies heavily on the geometric properties of the state difference set $\mathcal{D}_0(\rho)$. 
Here, we formally introduce the underlying mathematical framework, including semialgebraic sets and functions from real algebraic geometry.
Crucially, this framework allows us to invoke a {\L}ojasiewicz-type inequality, which rigorously quantifies the tame geometry of these sets and guarantees that boundaries cannot approach linear subspace intersections at a super-polynomially slow rate.

\begin{definition}[Semialgebraic sets]
A set $\cA\subset\mathbb{R}^n$ is called \emph{basic semialgebraic} if there exist polynomials
$p_1,\dots,p_r$, $q_1,\dots,q_s$, and $r_1,\dots,r_t$ in $\mathbb{R}[x_1,\dots,x_n]$ such that
\begin{align}
\cA=\bigl\{x\in\mathbb{R}^n\,\mid\, 
p_i(x)=0\ \forall i,\ \ q_j(x)\ge 0\ \forall j,\ \ r_k(x)>0\ \forall k\bigr\}.
\end{align}
If the strict inequalities are absent (i.e., $t=0$), then $\cA$ is called \emph{basic closed semialgebraic}.
A \emph{semialgebraic set} is a finite union of basic semialgebraic sets.
\end{definition}
\noindent While the definition extends to any real closed field, in this work we only consider subsets of $\mathbb{R}^n$.

A fundamental property of semialgebraic sets is their stability under projection:

\begin{proposition}[Tarski-Seidenberg, rephrased from Proposition 2.88 in~\cite{sm:BasuPollackRoy2006}]\label{sm:prop:tarski}
Let $\cA$ be a semialgebraic subset of $\mathbb{R}^{n+1}$ and 
$\Pi : \mathbb{R}^{n+1} \to \mathbb{R}^n$ the projection on the first $n$ coordinates.  
Then $\Pi(\cA)$ is a semialgebraic subset of $\mathbb{R}^n$.
\end{proposition}

We further define semialgebraic functions in terms of the semialgebraicity of their epigraphs.
\begin{definition}[Semialgebraic functions]
Let $\cA\subset\mathbb{R}^n$ be semialgebraic.
    A function $f:\cA\to\mathbb{R}$ is said to be \emph{semialgebraic} if its epigraph $\operatorname{epi}(f)\coloneqq\{(x,t)\,\mid\, x\in \cA,\ t\in\mathbb{R},\ t\ge f(x)\}$ is a semialgebraic subset of $\mathbb{R}^{n+1}$.
\end{definition}

We now show that restricting a semialgebraic function, or taking its sublevel sets, preserves its semialgebraicity.
\begin{lemma}[Restriction and sublevel sets of semialgebraic functions]\label{sm:lm:sublevel}
Let $f:\mathbb{R}^n\to\mathbb{R}$ be semialgebraic and let $A\subset\mathbb{R}^n$ be a semialgebraic set.
\begin{enumerate}[(i)]
\item The restriction $f|_A:A\to\mathbb{R}$ is semialgebraic.
\item For any $r\in\mathbb{R}$, the sublevel set $\{x\in\mathbb{R}^n\,\mid\,f(x)\le r\}$ is semialgebraic.
\end{enumerate}
\end{lemma}

\begin{proof}
(i) By definition, the epigraph,
\begin{gather}
    \operatorname{epi}(f|_A)=\{(x,t)\in A\times\mathbb{R}\,\mid\, t\ge f(x)\}
=\operatorname{epi}(f)\cap (A\times \mathbb{R}),
\end{gather}
is an intersection of semialgebraic sets, hence semialgebraic.

(ii) Consider the hyperplane $H_r\coloneqq\mathbb{R}^n\times\{r\}$, which is semialgebraic. Then
\begin{gather}
    \{x\,\mid\, f(x)\le r\}=\pi_n\bigl(\operatorname{epi}(f)\cap H_r\bigr),
\end{gather}
where $\pi_n:\mathbb{R}^{n+1}\to\mathbb{R}^{n}$ is the trivial projection. 
The intersection $\operatorname{epi}(f)\cap H_r$ is semialgebraic, and its projection is also semialgebraic by Proposition~\ref{sm:prop:tarski}.
\end{proof}

The most critical feature of semialgebraicity we employ in our robustness proof is the quantitative control it provides near common zero sets, formalized by the following inequality~\cite{sm:TaL1995,sm:Solerno1991}

\begin{theorem}[{\L}ojasiewicz inequality]\label{sm:thm:loja}
If $\cA\subset\mathbb{R}^n$ is a compact semialgebraic set and $f,g:\cA\to\mathbb{R}$ are continuous semialgebraic functions with $f^{-1}(0)\subset g^{-1}(0)$, then there exists a constant $c>0$ and a rational number $p\ge 1$ such that $|g(x)|^p\le c\,|f(x)|$ for any $x\in\cA$.
\end{theorem}

\section{Proof of universal robustness}
The UDA property is formulated for \emph{exactly} matching reduced density matrices.
In practice, however, marginals are obtained from finite data and are therefore inevitably noisy.
This raises a more general question regarding the unique determination: if a state $\sigma$ reproduces the marginals of a UDA state $\rho$ up to a small error, is $\sigma$ always close to $\rho$?
In this section, we provide a concrete derivation proving the universal robustness of UDA states.

\subsection{Semialgebraicity in Hermitian space}

To apply the semialgebraic tools introduced previously, we must verify that the relevant sets and functions used in our proof are semialgebraic.
Since these sets live in the complex space $\Herm(\cH)$ based on a $d$-dimensional Hilbert space $\cH$, we can fix a real-linear isomorphism $\Phi:\Herm(\cH)\to\mathbb{R}^{d^2}$ so that semialgebraicity is understood in the standard Euclidean space.

\begin{definition}[Vectorization of $\Herm(\cH)$]\label{sm:def:Phi}
Let $B$ be the Hilbert-Schmidt orthonormal basis of $\Herm(\cH)$,
\begin{gather}\label{sm:eq:basis_B}
B\coloneqq\{E_{ii}\}_{i=1}^d\ \cup\
\frac{1}{\sqrt{2}}\{E_{ij}+E_{ji}\}_{1\le i<j\le d}\ \cup\
\frac{1}{\sqrt{2}}\{\mathrm{i}\,(E_{ij}-E_{ji})\}_{1\le i<j\le d},
\end{gather}
where $E_{ij}$ denotes the matrix unit.
For any $X=\sum_{b\in B}x_b\,b\in\Herm(\cH)$, define the vectorization $\Phi(X)\coloneqq [x_b]_{b\in B}$.
\end{definition}

\begin{lemma}[Properties of the map $\Phi$]\label{sm:lm:Phi-iso-Lip}
The map $\Phi$ from Definition~\ref{sm:def:Phi} is a real-linear isomorphism.
Its inverse is the real-linear map $\Phi^{-1}$ such that $\Phi^{-1}\bigl([x_b]_{b\in B}\bigr)=\sum_{b\in B} x_b\, b$.
Equip $\Herm(\mathcal H)$ with the trace norm $\|\cdot\|_1$ and $\mathbb{R}^{d^2}$ with the Euclidean norm $\|\cdot\|_{\ell_2}$.
Then, both $\Phi$ and its inverse $\Phi^{-1}$ are Lipschitz continuous. 
\end{lemma}

\begin{proof}
Because $B$ constitutes an orthonormal basis for the real vector space of Hermitian matrices, every $X \in \Herm(\cH)$ admits a unique expansion $X = \sum_{b \in B} x_b b$ with real coefficients $x_b$. This uniqueness ensures $\Phi$ is well-defined and bijective, with an explicit inverse $\Phi^{-1}([x_b]_{b\in B}) = \sum_{b\in B} x_b b$. For any $X, Y \in \Herm(\cH)$ and scalars $\alpha, \beta \in \mathbb{R}$:
\begin{gather}
    \Phi(\alpha X + \beta Y) = \Phi\left(\sum_{b\in B} (\alpha x_b + \beta y_b) b\right) = [\alpha x_b + \beta y_b]_{b\in B} = \alpha\Phi(X) + \beta\Phi(Y).
\end{gather}
This establishes that $\Phi$ is a real-linear isomorphism.

For the Lipschitz continuity, fix $X,Y\in\Herm(\mathcal H)$ and set $Z\coloneqq X-Y=\sum_{b\in B}z_bb$.
With respect to the orthonormal basis $B$ regarding the Hilbert-Schmidt inner product, we have
\begin{gather}
\|\Phi(X)-\Phi(Y)\|_{\ell_2}=\|\Phi(Z)\|_{\ell_2}=\sqrt{\sum_{b\in B}z_b^2}=\|Z\|_{\mathrm{HS}}\leq\norm{X-Y}_1,
\end{gather}
where $\norm{\cdot}_{\mathrm{HS}}$ denotes the Schatten 2-norm (Hilbert-Schmidt norm), and $\norm{X}_\mathrm{HS}\leq\norm{X}_1$for all $X\in\Herm(\cH)$ .
Namely, $\Phi$ is Lipschitz and hence continuous.

For the inverse map $\Phi^{-1}$, let $x,y\in\mathbb{R}^{d^2}$ and set $z\coloneqq x-y$.
Since Cauchy-Schwarz inequality $\|X\|_1\le \sqrt{d}\|X\|_{\mathrm{HS}}$ holds for all $X\in\Herm(\mathcal H)$, we obtain
\begin{gather}
\norm{\Phi^{-1}(x)-\Phi^{-1}(y)}_1=\|\Phi^{-1}(z)\|_1\le \sqrt{d}\|\Phi^{-1}(z)\|_{\mathrm{HS}} = \sqrt{d}\|\Phi(\Phi^{-1}(z))\|_2 =\sqrt{d}\|x-y\|_2.
\end{gather}
\end{proof}

With this isomorphism, we can treat sets within Hermitian space as subsets of $\mathbb{R}^{d^2}$, where semialgebraicity is defined in the standard way.
Recall that $\mathcal D_0(\rho):=
\{\sigma-\rho:\sigma\in\mathcal D(\mathcal H)\}$ is the
state-difference set, and $\mathcal W_{\mathcal S}$ is the subspace of traceless Hermitian operators $X$ satisfying
$\operatorname{Tr}_{\bar S}(X)=0$ for every $S\in\mathcal S$.
Concretely, we define the coordinate images of these relevant sets under $\Phi$:
\begin{gather}\label{sm:eq:images}
    \cD'\coloneqq\Phi\!\bigl(\cD(\cH)\bigr),\qquad
    \cD_0'(\rho)\coloneqq\Phi\!\bigl(\cD_0(\rho)\bigr),\qquad
    \cW_{\cS}'\coloneqq\Phi\!\bigl(\cW_{\cS}\bigr).
\end{gather}

\begin{proposition}\label{sm:prop:semi-set}
    For any $\rho\in\cD(\cH)$ and set of subsystems $\cS$, all of $\cD'$, $\cD_0'(\rho)$, and $\cW'_{\cS}$ defined in Eq.~\eqref{sm:eq:images} are basic closed semialgebraic sets in $\mathbb{R}^{d^2}$. 
\end{proposition}
\begin{proof}
    For $\cD'$, the definition requires unit trace and positive semidefiniteness:
    \begin{gather}
        \cD'=\{x\in\mathbb{R}^{d^2}\,|\,\Phi^{-1}(x)\succeq0, \Tr \Phi^{-1}(x)=1\}.
    \end{gather}
    The unit trace constraint is a linear (hence polynomial) equation in the coordinates $x$.
    By Sylvester’s criterion~\cite{sm:horn2012matrix}, the condition $\Phi^{-1}(x) \succeq 0$ is equivalent to all $2^d$ principal minors of $\Phi^{-1}(x)$ being nonnegative. 
    Since $\Phi^{-1}(x)$ is a Hermitian matrix, its principal minors are determinants of Hermitian submatrices, which are inherently real numbers. 
    Because the entries of $\Phi^{-1}(x)$ are linear in $x$, these determinants are exactly real polynomials in $x$. 
    Thus, $\cD'$ is defined by finite polynomial equalities and inequalities, making it a basic closed semialgebraic set.
    The same logic applies to the shifted set $\cD'_0(\rho)$.

    For $\cW'_{\cS} = \ker(\cM_\cS \circ \Phi^{-1})$, the kernel of a finite-dimensional linear map is the common zero set of finitely many linear functions, which is inherently a basic closed semialgebraic set. 
\end{proof}

We next consider the functions used in our robustness estimates: the trace norm and the distance to the kernel subspace. Define the norm function $N(x) \coloneqq \|\Phi^{-1}(x)\|_1$ for $x\in\mathbb{R}^{d^2}$, and for a closed set $\Omega\subset\mathbb{R}^{d^2}$, the distance function $D_{\Omega}(x) \coloneqq \inf_{y\in\Omega} N(x-y)$.

\begin{proposition}\label{sm:prop:semi-func}
For any set of subsystems $\cS$, the norm $N(\cdot)$ and the distance $D_{\cW_\cS'}(\cdot)$ are semialgebraic functions.
\end{proposition}
\begin{proof}
Using the standard semidefinite programming (SDP) characterization of the trace norm, we have 
\begin{gather}
    \norm{X}_1=\min\{\Tr(U)+\Tr(V):\, X=U-V,\, U\succeq0,\,V\succeq0\}.
\end{gather}
We can represent the epigraph of $N$ by first defining the set:
\begin{gather}\label{sm:eq:epi_proj}
    \cT\coloneqq\{(x,u,v,t)\in\mathbb{R}^{d^2}\times \mathbb{R}^{d^2}\times \mathbb{R}^{d^2}\times \mathbb{R}\mid\,x=u-v,\,\Phi^{-1}(u)\succeq0,\,\Phi^{-1}(v)\succeq0,\,\Tr(\Phi^{-1}(u+v))\leq t\}.
\end{gather}
As established via Sylvester's criterion, the semidefinite constraints are polynomial inequalities.
Since trace and vector addition are linear operations, all constraints in $\cT$ are polynomial, making $\cT$ semialgebraic.
Let $\pi_{x,t}$ be the projection onto the $(x,t)$ coordinates. We have $\operatorname{epi}(N) =\{(x,t)\,\mid\,\norm{\Phi^{-1}(x)}_1\leq t\}= \pi_{x,t}(\cT)$. 
By the Tarski-Seidenberg principle in Proposition~\ref{sm:prop:tarski}, the projection of a semialgebraic set remains semialgebraic, so $N$ is a semialgebraic function.

As for the distance function, consider the set
\begin{gather}
    \cG \coloneqq \{(x,y,t)\in \mathbb{R}^{d^2}\times \mathbb{R}^{d^2}\times\mathbb{R}\mid\, y\in \cW_\cS',\, N(x-y)\le t\}.
\end{gather}
Since $\cW_\cS'$ is semialgebraic from Proposition~\ref{sm:prop:semi-set} and $N(x-y) \le t$ is simply a shifted epigraph condition of the semialgebraic function $N$, their intersection $\cG$ is semialgebraic.

Because $\cW_\cS'$ is a closed subspace, the infimum in the distance function can be attained.
Namely, we obtain
\begin{gather}
    t\ge D_{\cW_\cS'}(x)\quad\Longleftrightarrow\quad \exists\,y\in \cW_\cS'\ \text{such that}\ N(x-y)\le t.
\end{gather}
Consequently, $\operatorname{epi}(D_{\cW_\cS'})$ is the projection $\pi_{x,t}(\cG)$ onto $(x,t)$. 
By Proposition~\ref{sm:prop:tarski}, $\operatorname{epi}(D_{\cW_\cS'})=\pi_{x,t}(\cG)$ is semialgebraic. 
This completes the proof.
\end{proof}

\subsection{Proof of power-law robustness}
Here, we provide a concrete derivation proving the universal robustness of UDA states. 
For a UDA state $\rho$ with respect to the marginal support $\cS$, our proof relies on the distance of any traceless Hermitian matrix from the invisible subspace $\cW_\cS$: first, we bound this distance using the marginal norm of the matrix; second, we leverage semialgebraic geometry to relate this geometric distance to the total size when assuming the matrix is a valid state difference.

We begin by establishing the first ingredient.

\begin{lemma}\label{sm:lm:dis_margin}
    There exists a constant $C_1$ such that for every traceless Hermitian matrix $X\in\cV$,  
    \begin{gather}
        \dist(X,\cW_\cS)\leq C_1\norm{\cM_\cS(X)}_\cS= C_1\sum_{S\in\cS}\norm{X_S}_1.
    \end{gather}
\end{lemma}
\begin{proof}
Induced by the marginal map $\cM_\cS$ and its kernel (invisible subspace) $\cW_\cS$, we define the quotient space $\cV/\cW_\cS\coloneqq\{[X]= X+\cW_\cS\,\mid\,X\in\cV\}$ and define a corresponding quotient norm:
\begin{gather}
    \norm{[X]}_{\cV/\cW_\cS}\coloneqq\inf\{\norm{X-Y}_1:\,Y\in\cW_\cS\}=\dist(X,\cW_\cS),\quad\forall\,X\in\cV,
\end{gather}
    We construct the quotient version of the marginal map $\tilde\cM_\cS:\,\cV/\cW_\cS\rightarrow \Im(\cM_\cS)$ defined by
    \begin{gather}
        \tilde\cM_\cS([X])\coloneqq\cM_\cS(X).
    \end{gather}
    By definition, $\tilde\cM_\cS$ is a linear bijection, hence a linear isomorphism.
    
    Since both $\cV/\cW_\cS$ and $\Im(\cM_\cS)$ are finite-dimensional normed vector spaces, every bounded linear bijection has a bounded inverse~\cite{sm:kreyszig1978introductory}.
    Therefore, we obtain the following bound for any $X\in\cV$
    \begin{gather}
        \norm{[X]}_{\cV/\cW_\cS}=\norm{\tilde\cM_\cS^{-1}\circ\tilde\cM_\cS([X])}_{\cV/\cW_\cS}\leq\norm{\tilde\cM_\cS^{-1}}\norm{\tilde\cM_\cS([X])}_\cS\leq C_1\norm{\cM_\cS(X)}_\cS,
    \end{gather}
    where $\norm{\tilde\cM_\cS^{-1}}\coloneqq\sup_{Y\neq0}\frac{\norm{\tilde\cM_\cS^{-1}(Y)}_{\cV/\cW_\cS}}{\norm{Y}_\cS}$ is finitely bounded by $C_1$.
\end{proof}

With the geometric distance bounded by the marginal norm, the second step is to relate this distance to the global deviation $\|\delta\|_1$ from any state difference. 
Having established the semialgebraicity of the relevant sets and functions in the previous subsection, we can now invoke the {\L}ojasiewicz inequality to translate the qualitative UDA condition into a rigorous power-law separation bound.
Here $p\ge 1$ denotes the {\L}ojasiewicz exponent. In the robustness statements below, we use the equivalent exponent $\alpha\coloneqq1/p\in(0,1]$.

\begin{lemma}\label{sm:lm:Loja}
Let $\rho\in\cD(\cH)$ be a UDA state with respect to its marginal sets $\cS$. 
For every $r>0$, there exist a constant $C_2=C(r)>0$ and a rational number $p\ge 1$ such that for any $\delta\in \cD_0(\rho)$ with $\norm{\delta}_1\le r$,
\begin{equation}\label{sm:eq:sep}
\dist(\delta,\cW_\cS)\ \ge\ C_2\,\norm{\delta}_1^p.
\end{equation}
\end{lemma}
\begin{proof}
Since the state difference set $\cD_0(\rho)$ is closed and bounded, it is compact. 
The intersection $K \coloneqq \{\delta \in \cD_0(\rho)\,\mid\, \|\delta\|_1 \le r\}$ is a closed subset of a compact set, and thus is also compact. 
Under our isomorphism $\Phi$, its image $\Phi(K)$ is a compact semialgebraic set by Lemma~\ref{sm:lm:Phi-iso-Lip}.
Since $\Phi(K)$ equals $\cD'_0(\rho)\cap\{x\in\mathbb{R}^{d^2}\,\mid\, N(x)\leq r\}$, where both ingredients are semialgebraic sets, $\Phi(K)$ is semialgebraic according to Lemma~\ref{sm:lm:sublevel} and Proposition~\ref{sm:prop:semi-set}.

Consider the functions $g(x) \coloneqq N(x)$ and $f(x) \coloneqq D_{\cW_\cS'}(x)$ restricted to $\Phi(K)$. 
Both are semialgebraic and trivially continuous because norm and distance functions are originally Lipschitz continuous. 
Because $\rho$ is a UDA state, the only state difference in the kernel is the trivial one, i.e., $\cD_0(\rho) \cap \cW_\cS = \{\mathbf{0}\}$. 
Therefore, $f^{-1}(0) = \{\mathbf{0}\}$. Clearly, $g^{-1}(0) = \{\mathbf{0}\}$ as well.

Applying the {\L}ojasiewicz inequality in Theorem~\ref{sm:thm:loja} to $f$ and $g$ on $\Phi(K)$, there exists $c > 0$ and a rational number $p \ge 1$ such that $g(x)^p \le c f(x)$ for all $x \in \Phi(K)$. 
Substituting the definitions of $D_{\cW_\cS'}(x)$ and $N(x)$ yields 
\begin{gather}
    \|\delta\|_1^p \le c\cdot \dist(\delta, \cW_\cS),\quad\forall\,\delta\in K.
\end{gather}
Setting $C_2\coloneqq 1/c$ concludes the proof.
\end{proof}

Crucially, Lemma~\ref{sm:lm:Loja} provides a powerful bound, but it only holds within a bounded local neighborhood ($\|\delta\|_1 \le r$). To safely invoke it, we must guarantee that sufficiently small deviations of marginals actually trap the global deviation inside this small regime. 
The following lemma confirms this property.

\begin{lemma}[Local regime]\label{sm:lm:bootstrap}
Let $\rho\in\cD(\cH)$ be a UDA state with respect to its marginal sets $\cS$. 
For every $r>0$ there exists $d_0=d_0(r)>0$ such that for any $\delta\in\cD_0(\rho)$,
\begin{gather}
    \dist(\delta,\cW_\cS)< d_0\ \Longrightarrow\ \norm{\delta}_1<r.
\end{gather}
\end{lemma}

\begin{proof}
Define $K_r \coloneqq \{\delta \in \cD_0(\rho)\,\mid\, \|\delta\|_1 \ge r\}$.
If $K_r=\varnothing$, the conclusion holds for any $d_0>0$.
We may therefore assume $K_r\neq\varnothing$.
$\cD_0(\rho)$ is a closed and bounded subset of a finite‑dimensional Hilbert space, so it is compact.
Since $\{\delta\,\mid\, \norm{\delta}_1\ge r\}$ is closed, $K_r$ is a closed subset of a compact set, hence compact.
Moreover, since $K_r\cap \cW_\cS=\varnothing$ by the UDA condition, the continuous function $\dist(\cdot,\cW_\cS)$ has a positive minimum $d_0$ on $K_r$. 
Consequently, all $\delta\in\cD_0(\rho)$ satisfying $\dist(\delta,\cW_\cS)< d_0$ must obey $\norm{\delta}_1<r$.
\end{proof}

We now have all the necessary ingredients.
By chaining the observable marginal bound (Lemma~\ref{sm:lm:dis_margin}) and the geometric separation bound (Lemma~\ref{sm:lm:Loja}) with the local regime guarantee (Lemma~\ref{sm:lm:bootstrap}), we establish the universal robustness of UDA states.

\begin{theorem}[Power-law-robustness of UDA]\label{sm:thm:main}
Let $\rho\in\cD(\cH)$ be a UDA state with respect to its marginal sets $\cS=\{S_k\}_{k=1}^M$. 
Then there exist constants $C>0$, $\varepsilon_0>0$, and a rational number $\alpha\in(0,1]$ such that for every $\sigma \in \mathcal{D}(\mathcal{H})$ with
\begin{gather}\label{sm:eq:margin}
    \varepsilon=\norm{\mathcal{M}_{\mathcal{S}}(\sigma-\rho)}_{\mathcal{S}}\coloneqq\sum_{k=1}^M\norm{\Tr_{\bar{S_k}}(\sigma-\rho)}_1
\le
\varepsilon_0,
\end{gather}
we have
\begin{gather}
    \norm{\sigma - \rho}_1
\le
C \, \varepsilon^{\alpha}.
\end{gather}
\end{theorem}
\begin{proof}
Suppose $\sigma\in\cD(\cH)$ satisfies Eq.~\eqref{sm:eq:margin}.
Represent $\delta\coloneqq \sigma-\rho\in\cD_0(\rho)$, and the question reduces to bound $\norm{\delta}_1$.
By Lemma~\ref{sm:lm:dis_margin}, we have $\dist(\delta, \cW_\cS) \le C_1 \varepsilon$.

For the second step, we first fix $r>0$.
According to Lemma~\ref{sm:lm:bootstrap}, there exists a threshold $d_0>0$ such that $
\dist(\delta,\cW_\cS)< d_0$ ensures $\norm{\delta}_1<r$.
We can thus set $\varepsilon_0=d_0/(2C_1)$ with $C_1$ from Lemma~\ref{sm:lm:dis_margin}.
Based on this lemma and the condition in Eq.~\eqref{sm:eq:margin}, we get $\dist(\delta,\cW_\cS)\leq C_1\cdot\varepsilon\leq d_0/2<d_0$, which ensures $\norm{\delta}_1<r$ from Lemma~\ref{sm:lm:bootstrap}.
Consequently, with respect to the same $r$, we can apply Lemma~\ref{sm:lm:Loja} to $\delta$ since it satisfies the norm condition.
For $C_2$ from Lemma~\ref{sm:lm:Loja}, we have
\begin{align}
    \norm{\delta}_1\leq\left(\frac{\dist(\delta,\cW_\cS)}{C_2}\right)^{\alpha}\leq\left(\frac{C_1}{C_2}\right)^{\alpha}\varepsilon^{\alpha},
\end{align}
where $\alpha\coloneqq1/p$.
\end{proof}

Theorem~\ref{sm:thm:main} guarantees that all UDA states possess universal power-law robustness.
The exponent $\alpha$ naturally classifies the strength of this robustness, inducing a hierarchy where a larger $\alpha$ indicates stronger robustness, culminating in the optimal linear regime ($\alpha=1$). 
We formalize this classification as follows:

\begin{definition}[$\alpha$-robustness]\label{sm:def:linear_robust}
   A UDA state $\rho\in\cD(\cH)$ with respect to the marginal set $\mathcal S$ is said to be \emph{$\alpha$-robust} for $\alpha\in(0,1]$ if there exist constants $C>0$ and $\varepsilon_0>0$ such that for every $\sigma\in\mathcal D(\mathcal H)$ with
\begin{equation}
\label{sm:eq:robustness}
\varepsilon\coloneqq
\norm{\mathcal{M}_{\mathcal{S}}(\sigma-\rho)}_{\mathcal{S}}
\le
\varepsilon_0,
\end{equation}
we have
\begin{equation}
\label{sm:eq:robust-marginal}
\norm{\sigma - \rho}_1
\le
C\,\varepsilon^{\alpha}.
\end{equation} 
Particularly, we refer to $\alpha=1$ as linear robustness and $\alpha=1/2$ as square-root robustness.
\end{definition}

\section{Linear robustness certificate}
\label{sm:sec:uda-robust-cert}

By Theorem~\ref{sm:thm:main}, every UDA state is universally robust against deviations of marginals.
However, the exponent $\alpha$ is merely guaranteed to lie in $(0,1]$ and may be small.
In many applications, one seeks the best robustness, namely \emph{linear} robustness with $\alpha=1$, where the global deviations scale proportionally with the local deviations.
In this section, we derive a \emph{necessary and sufficient} condition for linear robustness.
Moreover, the condition can be cast through semidefinite-program feasibility problems, making linear robustness numerically testable.

\subsection{Tangent cone and the tangent criterion for linear robustness}
We formulate a geometric criterion for \emph{linear} robustness using the tangent cone, which captures the first-order feasible directions of a set at a given point.
Particularly, we use the Bouligand tangent cone, which is applicable for arbitrary subsets of a finite-dimensional normed space.

\begin{definition}[Tangent cone]
For a subset $\cA$ of $\Herm(\cH)$, we define its Bouligand tangent cone at $\mathbf{0}$ using sequences:
\begin{equation}
\label{sm:eq:tangent-def}
\cK_{\cA}(\mathbf{0})
\coloneqq
\left\{
X\in\Herm(\mathcal H)\,\mid\, \exists\,\{t_k,X_k\}_{k\in\mathbb N}\subset\mathbb{R}_+\times\cA\  \mathrm{s.t.}\  t_k\to 0^+\
\mathrm{and}\ 
\lim_{k\to\infty}\norm{X_k/t_k- X}_1=0
\right\}.
\end{equation}
\end{definition}

We adopt this concept for the state difference set $\cD_0(\rho)$ at the origin $\mathbf{0}$.
The following lemma shows that for this case, the abstract sequence definition simplifies into an explicit algebraic characterization.

\begin{lemma}[Tangent cone for $\cD_0(\rho)$]
\label{sm:lem:tangent-explicit-sec}
Consider a UDA state $\rho\in\cD(\cH)$.
Let $P_0$ be the projector onto $\ker\rho$. Then
\begin{equation}
\label{sm:eq:K-explicit-sec}
\cK_{\cD_0(\rho)}(\mathbf{0})=\{X\in\Herm(\mathcal H)\,\mid\,\Tr (X)=0,\ P_0XP_0\succeq 0\}.
\end{equation}
\end{lemma}

\begin{proof}
We prove this equality by mutual inclusion. Let $\cK \coloneqq \cK_{\cD_0(\rho)}(\mathbf{0})$.

For any $X\in \cK$, there exists a vanishing sequence with $t_k\to 0^+$ and $\delta_k\in \cD_0(\rho)$ such that $\delta_k/t_k\to X$.
Since for every $\delta_k\in\cD_0(\rho)$, $\Tr(\delta_k)=0$ for all $k$, ensuring $\Tr(X)=0$ by continuity of the trace.
Denote $\sigma_k\coloneqq \rho+\delta_k$, which is a valid quantum state by definition.
For any vector $v$ in $\ker\rho$, we have $\rho v=0$ and $\langle v,\sigma_k v\rangle\ge 0$, so $\langle v,\delta_k v\rangle\ge 0$.
Dividing $\langle v,\delta_k v\rangle$ by positive $t_k$ and taking the limit yields $\langle v, X v \rangle \ge 0$ for all $v \in \ker\rho$, which implies $P_0 X P_0 \succeq 0$.
Thus, we have $\cK \subseteq \{X : \Tr(X) = 0, \ P_0 X P_0 \succeq 0\}$

For the other side, consider any $X\in\Herm(\cH)$ satisfying $\Tr X=0$ and $P_0XP_0\succeq 0$.
Decompose the Hilbert space as $\mathcal H=\supp\rho\oplus\ker\rho$, and we choose the diagonalizing basis such that $\rho=\diag[\rho_{11},\mathbf{0}]$, where $\rho_{11}\succ0$ denotes the restriction of $\rho$
to its support $\operatorname{supp}\rho$.
In this block structure, $X$ takes the form
\begin{gather}
X=\begin{pmatrix}X_{11}&X_{10}\\[2pt]X_{01}&X_{00}\end{pmatrix},\text{ with } X_{00}\succeq 0.
\end{gather}
Since $\rho_{11}\succ0$, we have $A(t)=\rho_{11}+tX_{11}\succ 0$ for small $t>0$.
We define
\begin{gather}\label{sm:eq:rho_tilde}
\widetilde\rho(t)\coloneqq
\begin{pmatrix}
A(t) & tX_{10}\\[2pt]
tX_{01} & tX_{00}+t^2\,X_{01}A(t)^{-1}X_{10}
\end{pmatrix}.
\end{gather}
By construction, the Schur complement of $A(t)$ in $\widetilde\rho(t)$ is $tX_{00}\succeq0$, so  $\widetilde\rho(t)\succeq0$~\cite{sm:zhang2006schur}.
The trace is $\Tr\widetilde\rho(t)=1+t^2\beta(t)$, where $\beta(t)\coloneqq \Tr X_{01}A(t)^{-1}X_{10}$ is bounded for small $t$. Normalizing $\widetilde\rho(t)$ thus yields a valid state denoted by $\sigma(t)\in\cD(\cH)$.
We construct the state difference sequence $\{X_t \coloneqq \sigma(t) - \rho \in \cD_0(\rho)\}$.
Taking the limit as $t \to 0^+$ yields:
\begin{gather}
\lim_{t\to0^+}\frac{X_t}{t}
=
\lim_{t\to0^+}\frac{1}{\Tr(\widetilde\rho(t))}\cdot\frac{\widetilde\rho(t)-\rho}{t}
+\lim_{t\to0^+}\frac{1}{t}\left(\frac{1}{\Tr(\widetilde\rho(t))}-1\right)\rho
=X.
\end{gather}
Thus, $\{X \,\mid\, \Tr(X) = 0, \ P_0 X P_0 \succeq 0\} \subseteq \cK$, which completes the proof.
\end{proof}

With this equivalence, we arrive at the central theorem of this section: 

\begin{theorem}[Tangent criterion for linear robustness]
\label{sm:thm:lic-robust-marginal-op}
Let $\rho\in\cD(\cH)$ be a UDA state with respect to its marginal sets $\cS$. 
$\rho$ is linearly robust if and only if $\cK_{\cD_0(\rho)}(\mathbf{0})\cap \cW_\cS=\{\mathbf{0}\}$.
\end{theorem}

\begin{proof}
Let $\cK \coloneqq \cK_{\cD_0(\rho)}(\mathbf{0})$. Because $\cD_0(\rho)$ is convex and contains the origin, $\cD_0(\rho) \subseteq \cK$.

For the \enquote{\textbf{if}} side: 
Assume $\cK \cap \cW_\cS = \{0\}$. 
Define the unit slice of the cone as $S_\cK \coloneqq \{X \in \cK \,\mid\, \|X\|_1 = 1\}$, which is compact. 
Since the continuous function $\dist(\cdot, \cW_\cS)$ is strictly positive on $S_\cK$, it attains a minimum $\kappa > 0$.
Because $\cW_\cS$ is a linear subspace and $\cK$ is a cone, for any $X \in \cK$ we have:
\begin{gather}
\dist(X,\cW_\cS)=\norm{X}_1\,\operatorname{dist}\!\left(\frac{X}{\norm{X}_1},\cW_\cS\right)\ge \kappa\,\norm{X}_1.
\end{gather}

Since any state difference $\delta \in \cD_0(\rho)$ belongs to $\cK$, this geometric bound holds for all valid differences.
By our previous Lemma~\ref{sm:lm:dis_margin}, we already know $\dist(X, \cW_\cS) \le C_1 \|\cM_\cS(X)\|_\cS$ for all traceless Hermitian matrices $X\in\cV$. 
Combining these inequalities yields:
\begin{gather}
\norm{\delta}_1\leq\frac{1}{\kappa}\dist(\delta,\cW_\cS)\leq\frac{C_1}{\kappa}\norm{\cM_\cS(\delta)}_\cS,\ \forall\,\delta\in\cD_0(\rho)
\end{gather}
proving that global deviations are linearly bounded by deviations of local marginals.

For the \enquote{\textbf{only if}} side:
Assume $\rho$ is linearly robust with constant $C$, but suppose there exists a non-zero direction $X \in \cK \cap \cW_\cS$.
By definition, there exists a series $\{t_k,\delta_k\}$ with $\delta_k\in \cD_0(\rho)$ such that $t_k\to 0^+$ and $\lim_{k\to\infty}\norm{\delta_k/t_k- X}_1=0$.
This gives
\begin{gather}
    \lim_{k\to\infty}\norm{\mathcal M_{\mathcal S}(\delta_k)}_{\mathcal S}=\lim_{k\to\infty}t_k\norm{\mathcal M_{\mathcal S}(\delta_k/t_k)-\cM_{\cS}(X)}_{\mathcal S}= 0.
\end{gather}
Therefore, for any $\varepsilon_0>0$, there exists a large $k_0$ such that $\norm{\mathcal M_{\mathcal S}(\delta_k)}_{\mathcal S}\leq\varepsilon_0$ for all $k\geq k_0$.
The linear robustness thus generates $\norm{\delta_k}_1\le C\,\norm{\mathcal M_{\mathcal S}(\delta_k)}_{\mathcal S}$ for $k\geq k_0$.
Therefore, we have the following contradiction:
\begin{align}
   \norm{X}_1&=\lim_{k\to\infty}\norm{\frac{\delta_k}{t_k}}_1\leq\lim_{k\to\infty}\frac{C}{t_k}\norm{\mathcal M_{\mathcal S}(\delta_k)}_{\mathcal S}= \lim_{k\to\infty}C\norm{\mathcal M_{\mathcal S}(\delta_k/t_k)}_{\mathcal S}=C\norm{\mathcal M_{\mathcal S}(X)}_{\mathcal S}=0,
\end{align}
which is impossible since $X$ is non-zero.
\end{proof}

In fact, the tangent criterion can further reveal a sharp gap in the robustness hierarchy:
\begin{corollary}[Robustness gap]\label{sm:coro:gap}
    Let $\rho$ be a UDA state with respect to the marginal set $\cS$.
    If $\rho$ is not linearly robust, then $\rho$ is not $\alpha$-robust for any $1/2<\alpha<1$.
\end{corollary}
\begin{proof}
Let $\cK$ denote $\cK_{\cD_0(\rho)}(\mathbf{0})$.
For a full-rank $\rho$, we have $\cW_\cS=\{\mathbf{0}\}$, $\rho$ is always linearly robust.
Thus, we focus on rank-deficient $\rho$.

Suppose $\rho$ is not linearly robust, so there exists a nonzero $X\in\cK\cap \cW_\cS$.
We utilize the exact construction from Eq.~\eqref{sm:eq:rho_tilde}, and we construct a normalized state $\sigma(t) = \widetilde{\rho}(t) / (1 + t^2 \beta(t))$ such that 
\begin{align}\label{sm:eq:glo_trace_norm}
    \norm{\sigma(t)-\rho}_1\geq& t\norm{X}_1-\norm{\sigma(t)-\rho-tX}_1=t\norm{X}_1-\norm{\frac{\rho+tX+t^2R(t)}{1+t^2\beta(t)}-\rho-tX}_1\notag\\
    =&t\norm{X}_1-t^2\norm{\frac{R(t)-\beta(t)(\rho+tX)}{1+t^2\beta(t)}}_1=t\norm{X}_1-\order{t^2},
\end{align}
where $R(t)\coloneqq\diag[\mathbf{0},X_{01}A(t)^{-1}X_{10}]$ and $\beta(t)\coloneqq \Tr(X_{01}A(t)^{-1}X_{10})$.
For sufficiently small $t$, the linear term dominates, ensuring $\|\sigma(t) - \rho\|_1 \ge \frac{\|X\|_1}{2} t$.

Conversely, because $X \in \cW_\cS$ is invisible to the marginal map $\cM_\cS$, the first-order deviation vanishes on the marginal side. 
The deviation of marginals is purely driven by the second-order tail:
\begin{equation}
\label{sm:eq:upper-quadratic}
\varepsilon(t)\coloneqq\|\mathcal M_{\mathcal S}(\sigma(t)-\rho)\|_{\mathcal S}=\|\mathcal M_{\mathcal S}(\sigma(t)-\rho-tX)\|_{\mathcal S}
\le \|\mathcal M_{\mathcal S}\|\ \|\sigma(t)-\rho-tX\|_1
\le C_3 t^2,
\end{equation}
for some constant $C_3 > 0$.

Now fix any $\alpha\in(1/2,1)$ and suppose, for contradiction, that $\rho$ is $\alpha$-robust with constants $C$ and $\varepsilon_0$.
Choose $t>0$ small enough so that both the linear global deviation and Eq.~\eqref{sm:eq:upper-quadratic} hold and also $\varepsilon(t)\le \varepsilon_0$.
This robustness gives
\begin{gather}
    \frac{\norm{X}_1}{2}t \le \norm{\sigma(t)-\rho}_1\leq C\,\varepsilon(t)^{\alpha} \le C\,(C_3 t^2)^\alpha = C\,C_3^{\alpha}\,t^{2\alpha},
\end{gather}
which cannot hold given $t\to 0^+$.
\end{proof}

\subsection{Linear robustness certification}
By Lemma~\ref{sm:lem:tangent-explicit-sec} and Theorem~\ref{sm:thm:lic-robust-marginal-op}, linear robustness fails if and only if there exists a nonzero, traceless Hermitian matrix $X$ such that $\cM_{\cS}(X)=0$ and $P_0XP_0\succeq 0$, where $P_0$ is the projector onto $\ker(\rho)$.

Since $P_0XP_0\succeq 0$ implies either $P_0XP_0=0$ or $\Tr(P_0XP_0)>0$, we test the existence of such $X$ in two distinct steps.
First, we search for a nonzero solution with $P_0XP_0=0$, which reduces to the feasibility problem of a linear program:
\begin{equation}
\tag{$\mathbf P_L$}
\label{sm:eq:Pl-sec}
\text{find } X\in \Herm(\mathcal H)\backslash\{\mathbf{0}\}\ \ \text{s.t.}\
\Tr X=0,\ \ \mathcal M_{\mathcal S}(X)=0,\ \ P_0 X P_0= 0.
\end{equation}
If no such solution exists and $P_0\neq 0$, we then search for a direction $X$ with $\Tr(P_0XP_0)>0$.
Without loss of generality, we normalize this component to $\Tr(P_0XP_0)=1$, which reduces to an SDP feasibility problem
\begin{equation}
\tag{$\mathbf P_S$}
\label{sm:eq:PA-sec}
\text{find } X\in \Herm(\mathcal H)\ \ \text{s.t.}\
\Tr X=0,\ \ \mathcal M_{\mathcal S}(X)=0,\ \ P_0 X P_0\succeq 0,\ \ \Tr(P_0 X P_0)=1.
\end{equation}
If either step produces a feasible $X\neq 0$, the tangent criterion fails, and $\rho$ is \texttt{NOT ROBUST}.
Otherwise, $\rho$ is certified to be linearly robust.
We summarize the procedure in Algorithm~\ref{sm:alg:linear_certification}.
\begin{algorithm}[t]\label{sm:alg:linear_certification}
\caption{Linear robustness certification}
\KwIn{UDA state $\rho\in\mathcal D(\mathcal H)$; marginal set $\mathcal S$; marginal map $\mathcal M_{\mathcal S}$}
\KwOut{\texttt{ROBUST} or \texttt{NOT ROBUST}, and (if not robust) a witness $X\in \cK_{\cD_0(\rho)}(\mathbf{0})\cap \cW_\cS\setminus\{\mathbf{0}\}$}

\DontPrintSemicolon

$P_0 \leftarrow$ projector onto $\ker(\rho)$\;
\If{Program in~\eqref{sm:eq:Pl-sec} is feasible with $X$}{
  \Return \texttt{NOT ROBUST} with witness $X$\;
}

\eIf{the program in \eqref{sm:eq:PA-sec} is feasible with $X$}{
  \Return \texttt{NOT ROBUST} with witness $X$.\;
}{
  \Return \texttt{ROBUST}.\;
}
\end{algorithm}

\begin{theorem}\label{sm:thm:linear_certify}
    Let $\rho$ be a UDA state with respect to the marginal set $\cS$. 
    $\rho$ is linearly robust if and only if Algorithm~\ref{sm:alg:linear_certification} returns {\rm\texttt{ROBUST}}.
\end{theorem}
\begin{proof}
From the previous illustration and Lemma~\ref{sm:lem:tangent-explicit-sec}, $\texttt{ROBUST}$ arises iff both programs are infeasible, \emph{i.e.}, the tangent criterion holds.
\end{proof}

Linear robustness is a favorable yet stringent property that imposes strict informational requirements.
To successfully rule out all nontrivial invisible tangent directions, the marginal data must be sufficiently broad. 
This forces a resource trade-off: one must either access proportionally large subsystems or compensate by collecting an exponentially larger number of smaller marginals.
The following corollary quantifies this requirement.
\begin{corollary}[Lower bound on marginal size for linear robustness]
\label{sm:cor:lower-bound-s}
Let $\rho\in\cD(\cH)$ be an $n$-qubit UDA state of rank $r\leq d=\dim\cH=2^n$ with respect to the marginal set $\mathcal S=\{S_k\}_{k=1}^M$ with $s \coloneqq \max_k |S_k|$ denoting the maximum marginal size.
If $\rho$ is linearly robust, it is necessary that
\begin{equation}
\label{sm:eq:s-lower}
s\, \ge\, \frac{\log_2\!\bigl((r^2-1)+2r(d-r)\bigr)-\log_2 M}{2},
\end{equation}
Particularly, for pure states ($r=1$), this requires $s\,\ge\, (n-\log_2 M)/2$.
\end{corollary}

\begin{proof}
Linear robustness requires $\cK_{\cD_0(\rho)}(\mathbf{0})\cap \cW_\cS=\{\mathbf{0}\}$ by Theorem~\ref{sm:thm:lic-robust-marginal-op}.
Let $\cK\coloneqq\cK_{\cD_0(\rho)}(\mathbf{0})$ and $L\coloneqq\{X\in\Herm(\mathcal H)\,\mid\, \Tr X=0,\ P_0XP_0=0\}\subseteq \cK$ throughout the proof.
Since $L\subseteq\cK$, this forces $L\cap \cW_\cS=\{\mathbf{0}\}$.
On the other hand, $L\cap \cW_\cS=\ker(\mathcal M_{\mathcal S}|_L)$, so it means the restriction
$\mathcal M_{\mathcal S}|_L:L\to\bigoplus_{k=1}^M\Herm(\mathcal H_{S_k})$ is injective, which implies 
\begin{equation}
\label{sm:eq:dim-ineq}
\dim_{\mathbb{R}} L\ \le\ \sum_{k=1}^M \dim_{\mathbb{R}}\Herm(\mathcal H_{S_k})
\ =\ \sum_{k=1}^M (2^{|S_k|})^2
\ =\ \sum_{k=1}^M 4^{|S_k|},
\end{equation}
where $\dim_{\mathbb{R}}$ denotes the dimension as a real vector space.

Decompose $\mathcal H=\supp\rho\oplus\ker\rho$ with $\dim\supp\rho=r$ and $\dim\ker\rho=d-r$.
 In the diagonalizing basis where $\rho = \diag[\rho_{11}, \mathbf{0}]$, an operator $X \in L$ takes the block form:
\begin{gather}
X=\begin{pmatrix}A&B\\B^\ast&\mathbf{0}\end{pmatrix}.
\end{gather}
where $A \in \Herm(\cH_{\supp})$ is a traceless Hermitian matrix and $B \in \mathbb{C}^{r \times (d-r)}$ is an arbitrary complex matrix. 
Counting the independent real parameters gives:
\begin{gather}
\dim_{\mathbb{R}}L=(r^2-1)+2r(d-r).
\end{gather}
Substituting $\dim_{\mathbb{R}}L$ in \eqref{sm:eq:dim-ineq} and recalling that $|S_k|\le s$ for all $k$, we obtain Eq.~\eqref{sm:eq:s-lower}.
The pure-state case is directly achieved by setting $r=1$ and $d=2^n$.
\end{proof}

\section{Case studies}
In the preceding sections, we established a rigorous theoretical framework for the robustness of UDA states, culminating in a testable geometric criterion for linear robustness.
We now turn to applying this framework to representative UDA state families.

To facilitate this analysis, we first introduce a common mechanism for establishing UDA properties, originally derived in~\cite{sm:cramer2010efficient}.
Suppose a pure state is the unique ground state of a gapped Hamiltonian composed of local terms.
Then the deviation of marginals controls the energy gap, which in turn bounds the global deviations.
We rephrase this standard mechanism as a bridging proposition for the UDA robustness, which is proved in the main text.

\begin{proposition}\label{sm:prop:UniqueGround}
    Suppose the pure state $\rho$ is the unique ground state of an $n$-qubit Hamiltonian $H=\sum_{i=1}^m H_i$ with spectral gap $\Delta>0$ and local supports $\cS\coloneqq\{S_1,\cdots,S_m\subseteq[n]\}$.
    Then $\rho$ is a UDA state with respect to marginal supports $\cS$ such that for any $\sigma\in\cD(\cH)$
    \begin{gather}\label{sm:eq:quadra}
        \norm{\sigma-\rho}_1\leq 2\sqrt{\frac{\omega_{\max}}{\Delta}}\norm{\cM_\cS(\sigma-\rho)}_\cS^{1/2},
    \end{gather}
    where $\omega_{\max} \coloneqq\max_{i=1}^m(\lambda_{\max}(H_i)-\lambda_{\min}(H_i))/2$.
    
\end{proposition}

\subsection{Stabilizer states}
Let $\rho\coloneqq\ket{\psi}\bra{\psi}\in\cD(\cH)$ be an $n$-qubit stabilizer state with respect to a maximal stabilizer group $\textsf{Stab}(\psi) \subset {\sf P}^n$ (the Pauli group modulo phases).
We select independent generators $g_1,\cdots,g_n\in\textsf{Stab}(\psi)$, whose joint $+1$ eigenspace is the one-dimensional space $\textsf{span}(\ket{\psi})$.
 It is straightforward to show that $\ket{\psi}$ is uniquely determined by any marginal set $\cS$ that covers the supports of all $n$ generators. 
 Furthermore, relying on Proposition~\ref{sm:prop:UniqueGround}, we can immediately establish its baseline robustness.
\begin{proposition}[Square-root robustness of stabilizer states]
\label{sm:prop:stab-quadratic}
Let $\rho=\ket{\psi}\!\bra{\psi}\in\mathcal D(\mathcal H)$ be an $n$-qubit stabilizer state defined by independent Pauli generators $g_1, \dots, g_n$.
Let $\cS = \{S_k\}_{k=1}^n$ be the support set of these generators.
Then $\rho$ is at least square-root robustly UDA with respect to $\cS$.
\end{proposition}
\begin{proof}
    We construct the commuting parent Hamiltonian $H = \sum_{i=1}^n \Pi_i$, where the local projectors are defined as $\Pi_i \coloneqq (I - g_i)/2 \succeq 0$.
    By construction, $\ket{\psi}$ is the unique ground state of $H$ with an eigenvalue of $0$. 
    Because all other states must lie in the negative eigenspaces of at least one generator, the energy of any orthogonal state regarding $H$ is at least $1$, giving a spectral gap of $\Delta = 1$. 
    The eigenvalues of each local projector $\Pi_i$ are $0$ and $1$, yielding $\omega_{\max} = (1 - 0)/2 = 1/2$. Invoking Proposition~\ref{sm:prop:UniqueGround} directly ensures the square-root robustness bound:
    \begin{gather}
        \|\sigma - \rho\|_1 \le 2\sqrt{\frac{1/2}{1}} \|\cM_\cS(\sigma - \rho)\|_\cS^{1/2} = \sqrt{2} \|\cM_\cS(\sigma - \rho)\|_\cS^{1/2},\ \forall\, \sigma\in\cD(\cH).
    \end{gather}
\end{proof}

Besides this general baseline of at least square-root robustness, we can further show that it is generally easier to certify the linear robustness for all pure stabilizer states.
Specifically, we show that the SDP in \eqref{sm:eq:PA-sec} is automatically infeasible.
In this sense, we can focus exclusively on the linear program in~\eqref{sm:eq:Pl-sec}.

\begin{proposition}[Simplified linear certification]
\label{sm:prop:stab-linear-flat}
In the setting of Proposition~\ref{sm:prop:stab-quadratic}, define
\begin{gather}
L\coloneqq\{X\in\Herm(\mathcal H):\ \Tr X=0,\ P_0XP_0=0\},
\end{gather}
where $P_0$ is the projector onto $\ker(\rho)$.
The pure stabilizer state $\rho$ is linearly robust if and only if $L\cap \cW_\cS=\{0\}$. 
Consequently, the certification reduces entirely to checking the linear feasibility~\eqref{sm:eq:Pl-sec}.
\end{proposition}

\begin{proof}
Let $H=\sum_{i=1}^n \Pi_i$ be the parent Hamiltonian.
As established, $H$ has a ground state energy of $0$ and a spectral gap of $1$, which implies the operator inequality $H \succeq I - \rho = P_0$.

Let $\cK\coloneqq\cK_{\cD_0(\rho)}(\mathbf{0})$ throughout this proof.
Suppose $X\in \cK\cap \cW_\cS$ is a valid direction.
Since $X\in \cW_\cS$, it is invisible from marginals, hence
\begin{gather}
\Tr(HX)=\sum_{i=1}^n \Tr(\Pi_iX)=\sum_{i=1}^n \Tr\!\bigl(\Pi_i\,\Tr_{\overline{S_i}}X\bigr)=0.
\end{gather}
Write $X$ in block form related to  decomposition $\mathcal H=\supp\rho\oplus\ker \rho$ using the projectors $\rho$ and $P_0$:
\begin{gather}
X=\rho X\rho+\rho XP_0+P_0X\rho+P_0XP_0.
\end{gather}
Because $H\rho=0$ and $\rho H=0$, applying the trace reduces to the null-space block: $\Tr(HX) = \Tr(H P_0 X P_0)$.
Recall from Lemma~\ref{sm:lem:tangent-explicit-sec} that any $X\in \cK$ implies $P_0XP_0\succeq 0$.
Applying $H\succeq P_0$, we have
\begin{gather}
0=\Tr(HX)=\Tr\!\bigl(H\,P_0XP_0\bigr)\ \ge\ \Tr\!\bigl(P_0\,P_0XP_0\bigr)=\Tr(P_0XP_0)\ \ge\ 0.
\end{gather}
Consequently, $\Tr(P_0XP_0)=0$ and $P_0XP_0\succeq 0$ force $P_0XP_0=0$, i.e.\ $X\in L$.
Since $L\subseteq\cK$, this proves $\cK\cap \cW_\cS=L\cap \cW_\cS$.
By Theorem~\ref{sm:thm:lic-robust-marginal-op}, we can simplify the criterion.
\end{proof}

\subsection{Dicke states}
For any $n$-qubit system with $n\geq2$, the highly symmetric Dicke state with Hamming weight $1\leq k\leq n-1$ is defined as 
\begin{gather}
    \ket{D(n,k)}\coloneqq\binom{n}{k}^{-1/2}\sum_{\substack{\mathbf{x}\in\{0,1\}^n\\\textsf{wt}(\mathbf{x})=k}}\ket{\mathbf{x}}.
\end{gather}

We first establish that every Dicke state is at least square-root robust with respect to the full set of two-local marginals.
\begin{proposition}[Square-root robustness of Dicke states]
\label{sm:prop:dicke_quadratic}
Let $\rho_{n,k} \coloneqq \ket{D(n,k)}\bra{D(n,k)}$ for $1\leq k\leq n-1$. 
The state $\rho_{n,k}$ is a UDA state with respect to the set $\cS$ of all two-local marginals. 
Furthermore, it exhibits at least square-root robustness, satisfying the explicit bound:
\begin{gather}
    \norm{\sigma - \rho_{n,k}}_1 \le \sqrt{2 + \frac{2|n-2k|}{n-1}} \norm{\cM_\cS(\sigma - \rho_{n,k})}_\cS^{1/2},\ \forall\,\sigma\in\cD(\cH).
\end{gather}
\end{proposition}

\begin{proof}
To prove this, we construct a two-local parent Hamiltonian for which $\rho_{n,k}$ is the unique ground state:
\begin{gather}\label{sm:eq:H_nk}
H_{n,k}\coloneqq\sum_{1\le i<j\le n}\Pi^-_{ij}+\frac{1}{4}\,\left[\sum_{i=1}^nZ_i-(n-2k)I\right]^2,
\end{gather}
where $\Pi_{ij}^- \coloneqq (I - \text{SWAP}_{ij})/2$ projects onto the antisymmetric subspace of qubits $i$ and $j$.

We evaluate the spectrum of both terms. 
The projector $\Pi_{ij}^-$ can be written as $\frac{1}{4}(I - \vec{\sigma}_i \cdot \vec{\sigma}_j)$. 
Using the total angular momentum operator $J^2 = J_x^2 + J_y^2 + J_z^2$ where $J_\gamma \coloneqq \frac{1}{2}\sum_i \sigma_i^\gamma$, the sum over all pairs is:
\begin{equation}\label{sm:eq:S_as_J2}
S\coloneqq\sum_{1\leq i<j\leq n}\Pi^-_{ij}=\frac{n(n+2)}{8}I-\frac12\,J^2.
\end{equation}
The operator $J^2$ has eigenvalues $j(j+1)$ for total spin $j \in \{\frac{n}{2}, \frac{n}{2}-1, \dots\}$. 
The maximum spin $j = \frac{n}{2}$ corresponds exactly to the totally symmetric subspace, yielding an eigenvalue of $0$ for $S$. 
The next highest spin $j = \frac{n}{2}-1$ yields a first excited state energy of $\frac{n}{2}$. 
Thus, $S \succeq \frac{n}{2}(I - \Pi_{\text{sym}})$, where $\Pi_{\text{sym}}$ projects onto the totally symmetric subspace.

The second term in $H_{n,k}$ is diagonal in the computational basis. 
For a computational-basis state with Hamming weight $w$, the operator $\sum Z_i$ has eigenvalue $n-2w$. 
Thus, the squared bracket has eigenvalue $4(w-k)^2$. 
Scaled by $1/4$, this term vanishes strictly on the weight-$k$ subspace and has a smallest positive eigenvalue of $1$. Therefore, $\frac{1}{4}[\sum Z_i - (n-2k)I]^2 \succeq I - \Pi_k$, where $\Pi_k$ projects onto the weight-$k$ subspace.

Because each $\Pi_{ij}^-$ commutes with $Z_i + Z_j$, the operators $S$ and $[\sum Z_i]^2$ commute. 
Their common kernel is the intersection of the totally symmetric subspace and the weight-$k$ subspace, which uniquely defines the one-dimensional span of $|D(n,k)\rangle$. 
Every orthogonal state incurs an energy penalty from either the symmetry term or the weight term, yielding a strictly positive spectral gap:
\begin{gather}
    \Delta_{n,k}\geq\min\left\{\frac{n}{2},1\right\}=1.
\end{gather}

To invoke Proposition~\ref{sm:prop:UniqueGround}, we calculate the maximum spectrum range of $H_{n,k}$.
By distributing the one-local $Z_i$ contribution symmetrically among the two-local terms, it suffices to examine the pair $(1,2)$:
\begin{gather}
    H_{n,k}^{(1,2)}=\Pi_{12}^-+\frac{1}{2}Z_1Z_2-\frac{n-2k}{2(n-1)}(Z_1+Z_2),
\end{gather}
which has the spectrum range
\begin{gather}
    \omega_{12}\leq\frac{1}{2}+\frac{\abs{n-2k}}{2(n-1)}.
\end{gather}
\end{proof}

We now allow the marginal size to vary and determine the exact
threshold for linear robustness of the whole Dicke-state family. 

\begin{proposition}[Dicke marginal-size threshold]
\label{sm:prop:dicke-threshold}
Let $\rho_{n,k}=\ket{D(n,k)}\!\bra{D(n,k)}$, where $n\geq3$
and $1\leq k<n$, and define
\begin{gather}
\mathcal S_s\coloneqq\{S\subseteq[n]:|S|=s\},
\qquad
s_\star(n,k)\coloneqq\max\{k,n-k\}+\mathbf1_{\{2k=n\}}.
\end{gather}
For $2\leq s\leq n$, $\rho_{n,k}$ is UDA with respect to
$\mathcal S_s$, with
\begin{gather}
\alpha_\star(\rho_{n,k},\mathcal S_s)=
\begin{cases}
1/2,&s<s_\star(n,k),\\
1,&s\geq s_\star(n,k).
\end{cases}
\end{gather}
Moreover, no marginal collection whose largest support is smaller
than $s_\star(n,k)$ yields linear robustness.
\end{proposition}

\begin{proof}
Proposition~\ref{sm:prop:dicke_quadratic} establishes UDA and square-root robustness for all two-body marginals.
Partial-trace contractivity extends these conclusions to $\mathcal S_s$ for every $s\geq2$. 
It remains to characterize linear robustness.

By global bit-flip symmetry, assume $k\leq n/2$.
Let $P_0=I-\rho_{n,k}$.
The Hamiltonian $H_{n,k}$ in Eq.~\eqref{sm:eq:H_nk} satisfies $H_{n,k}\rho_{n,k}=0$ and $H_{n,k}\succeq\Delta_{n,k}P_0$, with $\Delta_{n,k}>0$.
According to Lemma~\ref{sm:lem:tangent-explicit-sec}, for any $X\in K_{\mathcal D_0(\rho_{n,k})}(0) \cap\mathcal W_{\mathcal S_s}$,
\begin{gather}
0=\operatorname{Tr}(H_{n,k}X) =\operatorname{Tr}(H_{n,k}P_0XP_0) \geq\Delta_{n,k}\operatorname{Tr}(P_0XP_0).
\end{gather}
Since $P_0XP_0\succeq0$, it must be 0. 
The linear criterion therefore reduces to preventing the following type of nonzero directions from being invisible
\begin{gather}
X=\ket{D(n,k)}\!\bra v+\ket v\!\bra{D(n,k)}
\in\mathcal W_{\mathcal S_s},
\qquad
\langle D(n,k)|v\rangle=0.
\end{gather}

\emph{Necessity.}
For $k<n/2$, consider the nonzero flat direction $\ket{D(n,k)}\!\bra{1^n}+\ket{1^n}\!\bra{D(n,k)}$.
Since the strings in $\ket{D(n,k)}$ can only be the same as $1^n$ at $k$ positions, the marginal maps that take partial traces on at least $k+1$ qubits must generate zero results.
Thus linear robustness requires largest support at least $n-k$.

For $n=2k$, we show that size-$k$ marginals also fail.
Choose a nonzero real function $f$ on the $k$-subsets of $[2k]$ satisfying
\begin{gather}
\sum_{\substack{R\subseteq Q\\|Q|=k}}f(Q)=0
\qquad \text{for every }R\subseteq[2k]\text{ with }|R|=k-1.
\end{gather}
This homogeneous system has $\binom{2k}{k}$ unknowns and $\binom{2k}{k-1}$ equations, so it admits a nonzero solution $f$.
We can define
$\ket v=\mathrm{i}\sum_{|Q|=k}f(Q)\ket Q$ where
$\ket Q$ has excitations exactly at the positions in $Q$ and correspondingly construct
\begin{gather}
X=\ket{D(2k,k)}\!\bra v+\ket v\!\bra{D(2k,k)}.
\end{gather}

We first show that this $X$ is traceless.
Fix $R\subseteq[2k]$ with $|R|<k$ and sum the constraints indexed by all $(k-1)$-subsets $J$ containing $R$.
The term $f(Q)$ appears once for every $J$ satisfying $R\subseteq J\subseteq Q$ and $|J|=k-1$.
Each such $J$ is obtained by deleting one element of $Q\setminus R$, so $f(Q)$ is counted exactly $k-|R|$ times.
Since $k-|R|>0$, we obtain
\begin{gather}
\sum_{\substack{Q\subseteq[2k]\\|Q|=k,\ R\subseteq Q}}f(Q)=0
\qquad
\text{for every }R\subseteq[2k]\text{ with }|R|<k.
\label{sm:eq:dicke-balanced-sums}
\end{gather}
Taking $R=\varnothing$ in
Eq.~\eqref{sm:eq:dicke-balanced-sums} gives
$\langle D(2k,k)|v\rangle=0$.
Thus $X$ is a nonzero flat tangent direction.

Now study the marginals of $X$.
Fix any subsystem $S\subseteq[2k]$ with $|S|=k$.
For each local excitation set $A\subseteq S$, define
\begin{align}
F_S(A)
\coloneqq&\sum_{\substack{Q\subseteq[2k]\\|Q|=k,\ Q\cap S=A}}f(Q)
\notag\\
=&\sum_{T\subseteq S\setminus A}(-1)^{|T|}
  \sum_{\substack{Q\subseteq[2k]\\|Q|=k,\ A\cup T\subseteq Q}}f(Q)
\notag\\
=&(-1)^{k-|A|}f(S),
\end{align}
where the second equality is inclusion--exclusion, and the last follows from Eq.~\eqref{sm:eq:dicke-balanced-sums}, since only the case of
$A\cup T=S$ has a nonzero summation.

For $A,A'\subseteq S$ with $|A|=|A'|$, the corresponding
matrix element of $X_S\coloneqq\operatorname{Tr}_{\bar S}X$ is
\begin{gather}
\langle A|X_S|A'\rangle
=
\frac{\mathrm{i}\bigl[F_S(A)-F_S(A')\bigr]}
{\sqrt{\binom{2k}{k}}}
=0.
\end{gather}
Matrix elements with $|A|\neq|A'|$ also vanish because both
$\ket{D(2k,k)}$ and $\ket v$ have total excitation number $k$.
Therefore every $k$-body marginal of $X$ vanishes, as does every
smaller marginal. This nonzero flat direction excludes linear
robustness for any marginal collection with largest support at
most $k$, proving the necessary condition $s\geq k+1$.

\emph{Sufficiency.}
Continue with $k\leq n/2$ and assume $s\geq\max\{k+1,n-k\}$.
By the preceding flat-space reduction, it suffices to prove that
\begin{gather}
X=\ket{D(n,k)}\!\bra v+\ket v\!\bra{D(n,k)}
\in\mathcal W_{\mathcal S_s},
\qquad
\langle D(n,k)|v\rangle=0
\end{gather}
implies $\ket v=0$.

Write $\ket v=\sum_{Q\subseteq[n]}v_Q\ket Q$, where $\ket Q$ has excitations exactly at the positions in $Q$.
For each $s$-subset $S$, let $X_S\coloneqq\operatorname{Tr}_{\bar S}X=0$.
For local excitation sets $A,A'\subseteq S$, expanding the
partial trace gives
\begin{align}
\langle A|X_S|A'\rangle
&=
\sum_{Z\subseteq\bar S}
\langle A\cup Z|X|A'\cup Z\rangle \notag\\
&=
\frac{1}{\sqrt{\binom nk}}
\left(
\sum_{\substack{Z\subseteq\bar S\\|Z|=k-|A|}}
\overline{v_{A'\cup Z}}
+
\sum_{\substack{Z\subseteq\bar S\\|Z|=k-|A'|}}
v_{A\cup Z}
\right)=0.
\label{sm:eq:dicke-flat-matrix-elements}
\end{align}
Each cardinality constraint comes from the weight-$k$ support
of the corresponding Dicke factor.
We use these equations to eliminate the coefficients sector by sector.

\emph{Step 1: Weights above $k$.}
Fix an integer $q$ with $k<q\leq n$.
We show that $v_Q=0$ for every $Q\subseteq[n]$ with $|Q|=q$.
Set $r\coloneqq\max\{0,q-s\}$.
For any $s$-subset $S$, choose $A,A'\subseteq S$ with
$|A|=q-r$ and $|A'|=k-r$.
These choices are valid because
$0\leq r\leq n-s\leq k$ and
$k<q-r=\min\{q,s\}\leq s$.
Since $|A|>k$, the first sum in
Eq.~\eqref{sm:eq:dicke-flat-matrix-elements} vanishes.
The second sum has $|Z|=r$, so it contains only weight-$q$
coefficients in $\{v_Q\}$:
\begin{gather}
0=\langle A|X_S|A'\rangle
=
\frac{1}{\sqrt{\binom nk}}
\sum_{\substack{Z\subseteq\bar S\\|Z|=r}}v_{A\cup Z}.
\label{sm:eq:dicke-high-weight-block}
\end{gather}

If $q\leq s$, then $r=0$. 
Choosing $A=Q$ in Eq.~\eqref{sm:eq:dicke-high-weight-block} gives $v_Q=0$.

If $q>s$, then $A=S$, and
Eq.~\eqref{sm:eq:dicke-high-weight-block} gives
\begin{gather}
\sum_{\substack{S\subseteq Q\\|Q|=q}}v_Q=0
\qquad\text{for every }|S|=s.
\label{sm:eq:dicke-inclusion-system}
\end{gather}
We will show that these equations force every $v_Q$ to vanish.
The idea is to compare supports differing at one site: subtracting their equations gives a smaller system for differences of coefficients.
To make this induction precise, we prove that Eq.~\eqref{sm:eq:dicke-inclusion-system} has only the zero solution for arbitrary complex coefficients $v_Q$ whenever
\begin{gather*}
0\leq s\leq q\leq n,\qquad s+q\geq n.
\end{gather*}
Our present parameters satisfy these conditions because $q>s$ and $s\geq n-k\geq n/2$.
We induct on $n-q$, allowing all parameters in this range.

\emph{Base case.}
If $n-q=0$, the only $q$-subset is $[n]$, and Eq.~\eqref{sm:eq:dicke-inclusion-system} gives $v_{[n]}=0$.

\emph{Induction step.}
Suppose $q<n$; then $s\geq n-q\geq1$.
Fix distinct sites $i,j$, and let $\Omega=[n]\setminus\{i,j\}$.
For any $(s-1)$-subset $R\subseteq\Omega$, compare the equations for
$S_i=R\cup\{i\}$ and $S_j=R\cup\{j\}$.
Separate the terms containing both $i,j$ from those containing only one of them:
\begin{align*}
0&=\sum_{\substack{R\subseteq T\subseteq\Omega\\|T|=q-1}}v_{T\cup\{i\}}
  +\sum_{\substack{R\cup\{i,j\}\subseteq Q\\|Q|=q}}v_Q,\\
0&=\sum_{\substack{R\subseteq T\subseteq\Omega\\|T|=q-1}}v_{T\cup\{j\}}
  +\sum_{\substack{R\cup\{i,j\}\subseteq Q\\|Q|=q}}v_Q.
\end{align*}
The second sums are identical and cancel upon subtraction.
Thus the coefficient differences
\begin{gather*}
w_T\coloneqq v_{T\cup\{i\}}-v_{T\cup\{j\}},
\qquad T\subseteq\Omega,\quad |T|=q-1,
\end{gather*}
satisfy
\begin{gather*}
\sum_{\substack{R\subseteq T\subseteq\Omega\\|T|=q-1}}w_T=0
\qquad\text{for every }R\subseteq\Omega\text{ with }|R|=s-1.
\end{gather*}
This is the same type of system on the $(n-2)$-element set $\Omega$, with parameters $(n-2,q-1,s-1)$.
They remain in the required range:
\begin{gather*}
0\leq s-1\leq q-1\leq n-2,\qquad
(s-1)+(q-1)\geq n-2.
\end{gather*}
Since the induction parameter is now $(n-2)-(q-1)=n-q-1$, the induction hypothesis gives $w_T=0$ for every $T$.

Hence replacing $i$ by $j$ in any $q$-subset leaves its coefficient unchanged.
The sites $i,j$ were arbitrary, and any two $q$-subsets can be connected by such replacements, so all coefficients share a common value $c$.
Substituting into Eq.~\eqref{sm:eq:dicke-inclusion-system} gives
\begin{gather*}
\binom{n-s}{q-s}\,c=0.
\end{gather*}
The binomial coefficient from the number of $Q$ sets covering $S$ is positive, so $c=0$, completing the induction.
Thus, $v_Q = 0$ for all $q > k$.

\emph{Step 2: Weights below $k$.}
Fix $Q\subseteq[n]$ with $q\coloneqq|Q|<k$.
Choose an $s$-subset $S\supseteq Q$ and a $k$-subset
$A\subseteq S$.
Equation~\eqref{sm:eq:dicke-flat-matrix-elements} gives
\begin{gather}
0=\sqrt{\binom nk}\,\langle A|X_S|Q\rangle
=
\overline{v_Q}
+
\sum_{\substack{Z\subseteq\bar S\\|Z|=k-q}}v_{A\cup Z}.
\end{gather}
The first term comes from the empty traced-out configuration.
Every coefficient in the second sum has weight
$|A|+|Z|=2k-q>k$ and therefore vanishes by Step~1.
Thus $v_Q=0$ for every $|Q|<k$.

\emph{Step 3: Weight $k$.}
Only weight-$k$ coefficients remain.
For any $k$-subsets $A,A'\subseteq S$, both sums in
Eq.~\eqref{sm:eq:dicke-flat-matrix-elements} contain only
$Z=\varnothing$, yielding
\begin{gather}
0=\langle A|X_S|A'\rangle
=\frac{v_A+\overline{v_{A'}}}{\sqrt{\binom nk}}.
\end{gather}
Taking $A=A'$ shows that each coefficient is purely imaginary.
The same equation then gives $v_A=v_{A'}$ whenever $A,A'$
lie in a common $s$-subset.
Since $s\geq k+1$, this holds for any two $k$-subsets
differing by a single-element replacement.
Such replacements connect all $k$-subsets, so their coefficients
share a common value $\mathrm{i}\lambda$, with
$\lambda\in\mathbb R$. Consequently,
\begin{gather}
\ket v
=
\mathrm{i}\lambda\sum_{|Q|=k}\ket Q
=
\mathrm{i}\lambda\sqrt{\binom nk}\ket{D(n,k)}.
\end{gather}
The condition $\langle D(n,k)|v\rangle=0$ forces
$\lambda=0$. Hence $\ket v=0$ and $X=0$.
The preceding flat-space reduction and the linear-robustness
criterion therefore establish linear robustness for
$s\geq s_\star(n,k)$.
\end{proof}

Setting $s=2$ recovers the complete two-body classification,
including the exceptional three-qubit states.

\begin{corollary}[Two-body Dicke robustness]
\label{sm:cor:dicke-two-body}
For $n\geq3$ and $1\leq k<n$, $\ket{D(n,k)}$ is exactly
square-root robust with respect to all two-body marginals,
except for $\ket{D(3,1)}$ and $\ket{D(3,2)}$, which are
linearly robust.
\end{corollary}

\begin{proof}
By Proposition~\ref{sm:prop:dicke-threshold}, linear robustness
holds precisely when $s_\star(n,k)\leq2$.
In the stated range, this occurs only for
$(n,k)=(3,1)$ and $(3,2)$; all remaining cases have
$\alpha_\star=1/2$.
\end{proof}

\subsection{A six-qubit UDA state beyond the UGS mechanism}
\label{sm:sec:six-qubit-robustness}

We now consider the six-qubit pure state introduced in~\cite{sm:karuvade2019uniquely},
which was proved to be UDA from all two-body marginals but not the
unique ground state of any two-local Hamiltonian.
We derive an explicit robustness bound for this state.

Label the qubits cyclically and define $\mathcal S_2\coloneqq\{\{i,j\}:1\leq i<j\leq6\}$.
We can decompose this two-body set into nearest neighbors and all others:
\begin{gather}
E_{\mathrm{nn}}\coloneqq\{\{1,2\},\{2,3\},\{3,4\},
                    \{4,5\},\{5,6\},\{6,1\}\},
\qquad
E_{\mathrm{far}}\coloneqq\mathcal S_2\setminus E_{\mathrm{nn}}.
\end{gather}
The target state is
\begin{gather}
\ket{\Psi_6}
=
\frac{1}{\sqrt2}
\left(
\ket{0^6}
+\frac13\sum_{\{i,j\}\in E_{\mathrm{far}}}\ket{ij}
\right),
\qquad
\rho_\Psi\coloneqq\ket{\Psi_6}\!\bra{\Psi_6},
\end{gather}
where $\ket{ij}$ has excitations only at sites $i,j$.
Ref.~\cite{sm:karuvade2019uniquely} proves that $\rho_\Psi$ is UDA with respect to
$\mathcal S_2$ but is not the unique ground state of any
two-local Hamiltonian.
Nevertheless, our Theorem~\ref{sm:thm:main} admits the robustness of this state without any further requirements besides the UDA property.
Moreover, we conduct a quantification on this robustness and analyzed the explicit scaling as follows.

\begin{proposition}[Robustness of the six-qubit UDA state]
\label{sm:prop:six-qubit-robustness}
For every six-qubit density matrix $\sigma$ with
$\varepsilon\coloneqq
\|\mathcal M_{\mathcal S_2}(\sigma-\rho_\Psi)\|_{\mathcal S_2}
\leq1$,
\begin{gather}
\|\sigma-\rho_\Psi\|_1
\leq C_\Psi\varepsilon^{1/4},
\qquad
C_\Psi\coloneqq
\sqrt2+2\sqrt{\frac{1+9\sqrt2}{3-\sqrt3}}.
\label{sm:eq:six-qubit-robustness}
\end{gather}
\end{proposition}

\begin{proof}
We first control leakage outside a subspace determined by the
nearest-neighbor marginals. Define
\begin{gather}
H_0\coloneqq
\sum_{\{i,j\}\in E_{\mathrm{nn}}}
\ket{11}\!\bra{11}_{ij},
\end{gather}
with identity operators on unlisted qubits understood.
Its kernel consists of computational-basis states with no
adjacent excitations:
\begin{gather}
\mathcal K
=
\operatorname{span}\left(
\{\ket{0^6}\}
\cup\{\ket i:i\in[6]\}
\cup\{\ket{ij}:\{i,j\}\in E_{\mathrm{far}}\}
\cup\{\ket{135},\ket{246}\}
\right).
\end{gather}
Let $P$ project onto this $18$-dimensional subspace.
Every basis state outside $\mathcal K$ contains an adjacent
excited pair, so $H_0\succeq I-P$.

Now consider an arbitrary six-qubit state $\sigma$.
Its projection onto the kernel space $\cK$ can be written as
\begin{gather}
\tau\coloneqq\frac{P\sigma P}{1-q},
\end{gather}
where $q\coloneqq\Tr[(I-P)\sigma]$.
This normalization can be further bounded by
\begin{align}
q=\operatorname{Tr}[(I-P)\sigma]
&\leq\operatorname{Tr}(H_0\sigma)=
\sum_{\{i,j\}\in E_{\mathrm{nn}}}
\operatorname{Tr}
\left[
\left(\ket{11}\!\bra{11}_{ij}-\frac12I_{ij}\right)\delta_{ij}
\right]
\notag\\
&\leq\frac12
\sum_{\{i,j\}\in E_{\mathrm{nn}}}\|\delta_{ij}\|_1
\leq\frac{\varepsilon}{2},
\label{sm:eq:six-qubit-leakage}
\end{align}
where $\rho_\Psi\in\cK$ and $\delta_{ij}\coloneqq\sigma_{ij}-(\rho_\Psi)_{ij}$ is traceless.
Since $q\leq1/2$, the projection is a valid state.
For any purification $\ket{\Omega}$ of $\sigma$, the vector
$(P\otimes I)\ket{\Omega}/\sqrt{1-q}$ purifies $\tau$
and has squared overlap $1-q$ with $\ket{\Omega}$.
The trace distance between these pure states is $2\sqrt q$.
Contractivity under partial trace therefore gives
\begin{gather}
\|\sigma-\tau\|_1\leq2\sqrt q.
\label{sm:eq:six-qubit-projection-error}
\end{gather}

We then need to bound the distance within $\mathcal K$.
To this end, we denote the two-local witness $W\coloneqq
\sum_{\{i,j\}\in E_{\mathrm{far}}}
\bigl(\ket{11}\!\bra{00}+\ket{00}\!\bra{11}\bigr)_{ij}$. 
Specifically,
its projection is
\begin{align}
PWP={}&
\sum_{\{i,j\}\in E_{\mathrm{far}}}
\bigl(\ket{0^6}\!\bra{ij}+\ket{ij}\!\bra{0^6}\bigr)
\notag\\
&+
\sum_{j\in\{1,3,5\}}
\bigl(\ket{135}\!\bra j+\ket j\!\bra{135}\bigr)
\notag\\
&+
\sum_{j\in\{2,4,6\}}
\bigl(\ket{246}\!\bra j+\ket j\!\bra{246}\bigr).
\label{sm:eq:six-qubit-compression}
\end{align}
These are three mutually orthogonal blocks.
For the first, consider the basis consisting of $\ket{0^6}$ and $\frac{1}{3}\sum_{\{i,j\}\in E_{\mathrm{far}}}\ket{ij}$.
This first block can be represented as 
\begin{gather*}
    \begin{pmatrix}
        0 &3\\
        3 &0
    \end{pmatrix}
\end{gather*}
in this subspace.
So it contributes $\pm3$ to the eigenvalues, and the positive one corresponds to $\ket{\Psi_6}$.
Each remaining block couples its three-excitation state
to the normalized sum of the corresponding single-excitation
states with amplitude $\sqrt3$.
All orthogonal combinations are annihilated.
Thus the nonzero eigenvalues are $\pm3$ and $\pm\sqrt3$,
with the latter pair occurring twice.
The unique largest eigenvalue is $3$, with eigenvector
$\ket{\Psi_6}$, giving
\begin{gather}
3P-PWP
\succeq
(3-\sqrt3)(P-\rho_\Psi).
\label{sm:eq:six-qubit-witness-gap}
\end{gather}
Each local term of $W$ has a two-body non-nearest-neighbor support and operator norm $1$. Hence
\begin{gather}
\left|\operatorname{Tr}[W(\sigma-\rho_\Psi)]\right|
\leq
\sum_{\{i,j\}\in E_{\mathrm{far}}}\|\delta_{ij}\|_1
\leq\varepsilon.
\label{sm:eq:six-qubit-witness-error}
\end{gather}

Taking the expectation of
Eq.~\eqref{sm:eq:six-qubit-witness-gap} in $\tau$, and using
$\operatorname{Tr}(W\rho_\Psi)=3$, yields
\begin{align}
(3-\sqrt3)
\bigl[1-\operatorname{Tr}(\rho_\Psi\tau)\bigr]
&\leq3-\operatorname{Tr}(W\tau) =
\operatorname{Tr}[W(\rho_\Psi-\sigma)]
+\operatorname{Tr}[W(\sigma-\tau)] \notag\\
&\leq\varepsilon+9\sqrt{2\varepsilon}.
\label{sm:eq:six-qubit-infidelity}
\end{align}
Since $\rho_\Psi$ is pure, the Fuchs-van de Graaf
bound and the triangle inequality imply
\begin{align}
\|\sigma-\rho_\Psi\|_1
&\leq
2\sqrt q+
2\sqrt{1-\operatorname{Tr}(\rho_\Psi\tau)} \notag\\
&\leq
\sqrt{2\varepsilon}
+
2\sqrt{
\frac{\varepsilon+9\sqrt{2\varepsilon}}{3-\sqrt3}
}
\notag\\
&\leq C_\Psi\varepsilon^{1/4}.
\end{align}
The last inequality uses
$\sqrt{\varepsilon}\leq\varepsilon^{1/4}$ and
$\varepsilon\leq\sqrt{\varepsilon}$ for
$0\leq\varepsilon\leq1$.
\end{proof}

\subsection{Exact robustness of QAOA states}
\label{sm:sec:qaoa-robustness}

We illustrate applicability beyond classically simulatable targets
using depth-one QAOA states, a class containing families with
conditional classical sampling hardness~\cite{sm:AbolinsAmbainis2026}.

Let $G=([n],E)$ be a graph of maximum degree at most three.
Consider the diagonal cost Hamiltonian
\begin{gather}
C=\sum_{i=1}^n a_iZ_i+
\sum_{\{i,j\}\in E}J_{ij}Z_iZ_j,
\qquad a_i,J_{ij}\in\mathbb R.
\end{gather}
With $B\coloneqq\sum_{i=1}^nX_i$ and real angles $\beta,\gamma$,
we define
\begin{gather}
U\coloneqq e^{-\mathrm{i}\beta B}e^{-\mathrm{i}\gamma C},
\qquad
\ket{\psi_U}\coloneqq U\ket{+}^{\otimes n}.
\end{gather}
Write $\rho_U\coloneqq\ket{\psi_U}\!\bra{\psi_U}$ and take
$\mathcal S_G\coloneqq\{S_i:i\in[n]\}$, where
$S_i\coloneqq\{i\}\cup N_G(i)$ and $N_G(i)$ denotes the neighbors
of vertex $i$.
Thus $\mathcal S_G$ contains at most $n$ distinct marginals,
each on at most four qubits due to the degree of three.

\begin{proposition}[Exact square-root robustness of QAOA states]
\label{sm:prop:qaoa-exact-robustness}
For $n\geq11$, $\rho_U$ is UDA with respect to
$\mathcal S_G$, with
$\alpha_\star(\rho_U,\mathcal S_G)=1/2$.
The robustness bound applies to arbitrary global density matrices.
\end{proposition}

\begin{proof}
Set $h_i\coloneqq U(I-X_i)U^\dagger/2$ and
$H_U\coloneqq\sum_{i=1}^n h_i$.
The Hamiltonian $H_U$ is unitarily equivalent to
$\sum_i(I-X_i)/2$, so it has unique ground-state projector
$\rho_U$, ground energy zero, and spectral gap one.

To determine the support of $h_i$, note that all terms of
$C$ commute. Their exponential factors can therefore be
reordered so that those not involving $i$ act directly on
$(I-X_i)/2$ and cancel with their adjoints.
Only the factors generated by $Z_i$ and $Z_iZ_j$ with
$j\in N_G(i)$ remain.
They act within $S_i$, and the subsequent single-qubit
mixer rotations do not enlarge this support.
Thus $h_i$ is supported on $S_i$.
Grouping terms with coincident supports,
Proposition~\ref{sm:prop:UniqueGround} gives UDA and
$\alpha^\star\geq1/2$.

By Corollary~\ref{sm:coro:gap}, establishing optimality reduces to excluding linear robustness.
For a pure $n$-qubit state, Corollary~2 requires
\begin{gather}
2(2^n-1)\leq M4^s,
\end{gather}
where $M:=|\mathcal S_G|$ and
$s:=\max_{S\in\mathcal S_G}|S|$.
Here $M\leq n$ and $s\leq4$, whereas
$n4^4<2(2^n-1)$ for $n\geq11$.
Linear robustness is therefore impossible, so
$\alpha^\star\leq1/2$.
Together with the UGS bound, this proves the claim.
\end{proof}

To connect this classification with classical simulability, Corollary~11 of Ref.~\cite{sm:AbolinsAmbainis2026} shows that a polynomial-time classical algorithm sampling the computational-basis output distribution of every efficiently specified circuit in this class within multiplicative error $2^{n^\eta}$, for some fixed $0\leq\eta<1$, would collapse the polynomial hierarchy to its third level.
Our robustness framework and this analysis based on it therefore give the exact robustness information for a UDA class with conditional
sampling hardness.

\end{document}